\documentclass[runningheads]{llncs}
\usepackage[utf8]{inputenc}
\usepackage[english]{babel}
\usepackage{lmodern}
\usepackage[T1]{fontenc}
\usepackage[babel=true]{microtype}
\usepackage[
  left=1in,
  right=1in,
  top=1in,
  bottom=1in
]{geometry}
\usepackage{latexsym}
\usepackage{amsmath,amssymb}
\usepackage{fontawesome5}
\usepackage[numbers]{natbib}
\usepackage{bm}
\usepackage{booktabs}
\usepackage{enumitem}
\usepackage{graphicx}
\usepackage{caption}
\usepackage{amsfonts}
\usepackage{tikz}
\usepackage{mathtools}
\usepackage{multicol}
\usepackage{subcaption}
\usepackage{dsfont}
\usepackage{hyperref}
\usepackage[noend,ruled]{algorithm2e}
\usepackage{cleveref}

\newcommand{\Oh}{\mathcal{O}}
\newcommand{\W}{\widehat{W}}
\newcommand{\adj}[3]{#1^{#3}_{#2}}
\newcommand{\righty}[2]{\adj{#1}{#2}{+}}
\newcommand{\lefty}[2]{\adj{#1}{#2}{-}}
\DeclarePairedDelimiter\ceil{\lceil}{\rceil}
\DeclarePairedDelimiter\floor{\lfloor}{\rfloor}
\newcommand{\ignore}[1]{}
\newcommand{\point}{{\textsc{Point}}}
\newcommand{\interval}{\textsc{Int}}

\SetKw{Continue}{continue}

\SetKwIF{If}{ElseIf}{Else}{if}{:}{else if}{else}{end if}%
\SetKwFor{For}{for}{:}{}

\usepackage{tikz}
\usetikzlibrary{backgrounds}
\usetikzlibrary{arrows}
\usetikzlibrary{shapes,shapes.geometric,shapes.misc}

\tikzstyle{tikzfig}=[baseline=-0.25em,scale=0.5]

\pgfkeys{/tikz/tikzit fill/.initial=0}
\pgfkeys{/tikz/tikzit draw/.initial=0}
\pgfkeys{/tikz/tikzit shape/.initial=0}
\pgfkeys{/tikz/tikzit category/.initial=0}

\pgfdeclarelayer{edgelayer}
\pgfdeclarelayer{nodelayer}
\pgfsetlayers{background,edgelayer,nodelayer,main}

\tikzstyle{none}=[inner sep=0mm]

\newcommand{\tikzfig}[1]{%
{\tikzstyle{every picture}=[tikzfig]
\IfFileExists{#1.tikz}
  {\input{#1.tikz}}
  {%
    \IfFileExists{figures/#1.tikz}
      {\input{figures/#1.tikz}}
      {\tikz[baseline=-0.5em]{\node[draw=red,font=\color{red},fill=red!10!white] {\textit{#1}};}}%
  }}%
}

\newcommand{\ctikzfig}[1]{%
\begin{center}\rm
  \input{#1.tikz}
\end{center}}

\tikzstyle{every loop}=[]

\tikzstyle{circle}=[fill=black, draw=black, shape=circle]
\tikzstyle{empty circle}=[fill=white, draw=black, shape=circle]
\tikzstyle{diamond}=[fill={rgb,255: red,213; green,213; blue,213}, draw=black, shape=rectangle, aspect=4, scale=0.4]
\tikzstyle{rect}=[fill=white, draw=black, shape=rectangle]
\tikzstyle{selected}=[fill={rgb,255: red,81; green,181; blue,0}, draw=black, shape=rectangle]

\tikzstyle{normal}=[-, fill=none]
\tikzstyle{dotted}=[-, dashed]
\tikzstyle{pointer}=[->]
\tikzstyle{measure}=[<->, draw={rgb,255: red,0; green,46; blue,255}]
\tikzstyle{probspace}=[draw={rgb,255: red,2; green,82; blue,255}, <->]
\tikzstyle{Wcandidate}=[->, draw={rgb,255: red,28; green,113; blue,0}]
\tikzstyle{nonWcandidate}=[draw={rgb,255: red,211; green,0; blue,0}, ->]
\tikzstyle{filled}=[-, fill={rgb,255: red,162; green,162; blue,162}, draw=none, tikzit draw={rgb,255: red,6; green,6; blue,6}]
\tikzstyle{measuredotted}=[draw={rgb,255: red,0; green,46; blue,255}, dashed, <->]
\tikzstyle{measuresingle}=[draw={rgb,255: red,0; green,46; blue,255}, ->]
\tikzstyle{dotted measure}=[-, draw={rgb,255: red,0; green,53; blue,245}, dashed]
\tikzstyle{pointerdotted}=[->, dashed, draw=red]

\title{Multiwinner Voting with Interval Preferences \\ under Incomplete Information}

\author{Drew Springham\inst{1,3} \and Edith Elkind\inst{2} \and Bart de Keijzer\inst{1} \and Maria Polukarov\inst{1}}
\institute{King's College London\and Northwestern University\and \email{drew.springham@kcl.ac.uk}}

\begin{document}



\maketitle 

\begin{abstract}
In multiwinner approval elections with many candidates, voters may struggle to determine their preferences over the entire slate of candidates. It is therefore of interest to explore which (if any) fairness guarantees can be provided under reduced communication. 
In this paper, we consider voters with one-dimensional preferences: voters and candidates are associated with points in $\mathbb R$, and each voter's approval set forms an interval of $\mathbb R$.
We put forward a probabilistic preference model, where the voter set consists of $\sigma$ different groups; each group is associated with a distribution over an interval of $\mathbb R$, so that 
each voter draws the endpoints of her approval interval from the distribution associated with her group. We present an algorithm for computing committees that provide {\em Proportional Justified Representation + (PJR+)}, which proceeds by querying voters' preferences, and show that, 
in expectation, it makes $\mathcal{O}(\log( \sigma\cdot k))$ queries per voter, where $k$ is the desired committee size.
\end{abstract}

\section{Introduction}
The problem of selecting a set of winning candidates from a wider pool, given the preferences of many agents, arises in a wide variety of settings. This problem is formalized as \emph{Multiwinner Voting (MWV)}. Real-world examples include, but are not limited to, committee selection, parliamentary elections, group recommendations, and opinion summarization \citep{Lackner2023Multi-WinnerPreferences,Halpern2024ComputingVotes,revel25}. Often, we wish to select a committee of size $k$ in a way that is ``fair''; this is captured 
by the concept of \emph{Justified Representation (JR)} and its extensions, such as PJR, PJR+, EJR, EJR+, FPJR, and FJR
\citep{Sanchez-Fernandez2017ProportionalRepresentation,Brill2023RobustVoting,Aziz2017JustifiedVoting,Peters2020ProportionalityWelfarism,Kalayci2025FullRepresentation}.
Each of these extensions captures the idea that groups making up at least an $\ell/k$ fraction of the voter population deserve to be represented by at least $\ell$ of the $k$ candidates in the winning set. For example, for the winning set $W$ to provide PJR+\,---the axiom we focus on in this work---for every group of voters $S$ that constitutes an $\ell/k$ fraction of the population and approves of a candidate not in $W$, the set $W$ must contain $\ell$ candidates approved by some voter in $S$.

Importantly, some multiwinner elections have a large number of candidates in consideration. 
E.g., in the well-known \emph{participatory budgeting} setting, in which residents of a city can propose projects to be completed, and then vote on the proposals, there may be dozens of projects to choose from: indeed, 
there are real-life examples of participatory budgeting elections with more than 150 candidate projects \citep{Faliszewski2023ParticipatoryAnalysis}.
Another example is during debates on Pol.is \citep{small2021polis}, an online platform for civic participation. Here, citizens vote on comments (candidates) that are made by others, and we wish to select a representative slate of opinions. In some debates, e.g. during an online citizens assembly meeting on climate law in Austria, there can be over 1000 comments made \citep{TheComputationalDemocracyProject2022TheAustria}.
In such circumstances, one cannot realistically expect the voters to accurately evaluate every single candidate. 
Therefore, it is of interest to consider multiwinner voting mechanisms that do not require full information from each voter about their preferences, but still provide guarantees on the social desirability of the outcome.

In this paper, we introduce the general problem of multiwinner voting with incomplete information in a spatial preferences setting, where voters and candidates are associated with points in ${\mathbb R}^d$, and each voter approves candidates within a certain distance from her. We focus on the one-dimensional case of this general problem, where each voter's approval set forms an interval of $\mathbb R$ (her \emph{approval interval}). 
Such a model may be appropriate, e.g., in a political setting, where we can view candidates as positioned on the left-right political spectrum; then, each voter only disapproves of candidates they find too extreme (to the left or to the right). 

We put forward a rich probabilistic model to represent the voters' preferences in this setting: the \emph{Random Interval Voter model (RIV)}. In RIV (\Cref{def:sriv}), we assume that voters are divided into groups, so that each group is associated with a disjoint subinterval of $\mathbb R$ and a voter's approval interval is obtained by independently sampling two points from a given distribution over the interval associated with their group. 

When the voter population is large, a probabilistic model can be thought of as providing aggregated preference information, based, e.g., on historical data.
In particular, an attractive feature of RIV is that it can express both large-scale trends and small-scale behaviour.
Indeed, by associating each voter with a disjoint subinterval of $\mathbb R$, we can capture settings where candidates are partitioned into well-defined groups (parties), so that voters' preferences are consistent with this partition; in this setting, each party naturally corresponds to a subinterval of $\mathbb R$. Thus, the RIV model can be thought of as a probabilistic variant of open party-list elections.
By modelling voters as having random endpoints within each segment, we can capture voters' within-party preferences by specifying a segment-specific cumulative distribution function. 
Together, both factors give RIV sufficient depth to be a nuanced model of how agents might vote in practice.


\paragraph{Our contributions.}
We develop a preference elicitation framework that is tailored to the RIV model. Our framework allows two types of queries: {\em point queries}, where we ask if a voter approves a specific candidate, and {\em interval queries},  
where we ask if a voter's set of approved candidates is contained within a given interval. We show that, in this model, we can elicit the preferences of each voter using $\Oh(\log m)$ queries per voter in expectation, where $m$ is the number of candidates; this enables us to construct an outcome in the core \cite{Aziz2017JustifiedVoting,Pierczynski2022Core-StableDomains}.
Then, our main technical contribution is an algorithm that outputs committees providing a notion of proportionality that is slightly weaker than core stability, namely, PJR+ \cite{Brill2023RobustVoting}, and only requires $\mathcal{O}( \log (\sigma k))$ queries per voter in expectation, 
where $k$ is the target committee size and $\sigma$ the number of segments. Notably, unlike for the core, our bound for PJR+ does not depend on the total number of candidates.
Moreover, while the bound on the running time holds in expectation (and also with high probability), the PJR+ guarantee holds with probability~1.

\paragraph{Outline.} In \Cref{sec:related}, we review prior work on multiwinner voting with incomplete information and preference elicitation. Next, in \Cref{sec:prelims} we introduce the preliminary definitions and results, which lay the foundation for the contributions of this paper. In \Cref{sec:ciriv}, we describe our probabilistic model, and formulate the general problem of multiwinner voting with incomplete information in a spatial setting. \Cref{sec:obt} discusses our preferences elicitation framework. We also formulate accompanying lemmas, which we put to use in \Cref{sec:pjr}, where we prove that we can construct a committee that provides PJR+ using limited queries per voter in expectation. Finally, in \Cref{sec:conclusion} we conclude with ideas for future work and open questions.

\section{Related Work}
\label{sec:related}
There is substantial literature on approval-based multiwinner voting in general (see \citep{Lackner2023Multi-WinnerPreferences} for an overview). 
The preference elicitation aspect of voting has also been studied extensively, especially in the single-winner context (e.g. \citep{Conitzer2002VoteStrategy-Proofness,Meir2014AEquilibria,Dey2015SampleElections,Dery2016ReducingMaking,Zhao2018AAggregation}).

A number of papers on preference elicitation in multiwinner voting  optimise objectives other than proportionality, e.g., 
minimax regret \citep{Lu2013Multi-WinnerPreferences}, diversity \citep{Lindeboom2025AInformation}, or social welfare~\citep{Mandal2020OptimalVoting}.

\citet{Imber2022Approval-BasedInformation} consider the problem of winner determination with incomplete preferences, under several voting rules  
that guarantee forms of justified representation. However, their model does not 
allow for preference elicitation: In their setting, each voter reports a subset of their approved candidates and a subset of their disapproved candidates, with the remaining candidates being unknown, and the goal is to determine which candidates are possibly or necessarily in the winning set.
Subsequently, \citet{Imber2024SpatialInformation} considered the same problem under the assumption that voters and candidate lie in a $d$-dimensional Euclidean space. There are two crucial differences between this line of work and our contribution: (1) we allow the algorithm to query the voters, and (2) we take a probabilistic perspective rather than a nondeterministic perspective. 

The closest in spirit to our approach is the work of 
\citet{Halpern2023RepresentationVotes}, who also consider querying voters in order to achieve a form of justified representation.  
They formalise an idea of approximate justified representation, by scaling up the size of a voter group that ``deserves'' representation by a factor of $(1+\varepsilon)$, and provide adaptive algorithms to achieve approximate EJR (a notion that is similar to PJR+, but incomparable with it). They also provide a lower bound on the number of voters that need to be sampled to achieve EJR with probability of at least $1-\delta$. \citet{Brill2023RobustVoting} discuss a similar technique for achieving EJR+ with high probability in the general preference setting. 
Our work differs from \citet{Halpern2023RepresentationVotes} and \citet{Brill2023RobustVoting} in two main ways: 
(1) while their model allows for unrestricted approval preferences, we assume that voters and candidates are located in a metric space; (2) our model is Bayesian, i.e., our algorithm makes use of distributional information and our bounds on the number of queries hold with high probability and in expectation. 
In particular, our query model allows for queries that
ask a voter if her approval set is fully contained within a given interval; 
such queries are very natural in the metric setting, but their non-metric analogue (asking a voter whether all her approved candidates belong to some subset of candidates) has a higher cognitive and communication cost.
The metric and distributional assumptions, together with the richer query model, enable us to obtain stronger results: while the number of queries in the work of \citet{Halpern2023RepresentationVotes} scales linearly with the number of candidates $m$, our algorithm achieves PJR+ using a number of queries that, in expectation, does not depend on $m$.
Moreover, we achieve proportionality with probability 1, whereas \citet{Halpern2023RepresentationVotes} and \citet{Brill2023RobustVoting} aim for EJR+ with probability $1-\varepsilon$.


\section{Preliminaries}
\label{sec:prelims}
Here, we present the relevant background to introduce our work.
All our indexing of ordered sets (or lists) starts at 0, and we use ``list slicing'' notation of the form $X[a:b]\vcentcolon=\{X[i]:a\leq i <b\}$.\footnote{Similar to the Python programming language.} We also use the notation $[t]\vcentcolon=\{i\in \mathbb{N}:1\leq i \leq t\}$ to denote the set of natural numbers up to $t$.

\begin{definition}{(Approval-based MWV election)}
An \emph{MWV election} is a tuple $E=(V,C,k,A)$, where $V$ is a set of $n$ voters, $C$ is a set of $m$ candidates, $k\in \mathbb{N}$ is a total number of candidates to elect, and $A:V\rightarrow 2^C$ is a function that maps each voter to the set of candidates she approves.
For each $G\subseteq V$, we write $A(G)=\bigcap_{v\in G}A(v)$.
The primary task associated with an MWV election is to select a set, or committee, of {\em winners} $W\subseteq C$ with $|W|=k$.
\end{definition}
The literature on multiwinner voting with approval preferences defines a number of fairness axioms, ranging from Justified Representation (JR) to core stability~\cite{Lackner2023Multi-WinnerPreferences}. 
All these axioms aim to capture the idea that large groups of voters with similar preferences should be represented in the committee in proportion to the group size;
however, they differ in how they define which groups of voters are considered to have similar preferences and what it means to represent a group.
Below, we define two axioms from this hierarchy, namely, PJR+ \cite{Brill2023RobustVoting} and core stability~\cite{Aziz2017JustifiedVoting}. In what follows, we will show that our method selects PJR+ committees with high probability while asking a small number of queries; obtaining a similar result for more demanding axioms such as EJR+ remains a direction for future work, discussed in \Cref{sec:conclusion}.

\begin{definition}{(PJR+)}
A committee $W$ satisfies PJR+ \citep{Brill2023RobustVoting} (a.k.a.~IPSC \citep{Aziz2021ProportionallyPreferences}) if $|W|=k$ and for every $\ell \in [k]$ and every group $G\subseteq V$ that satisfies $|G|\geq n\ell /k$, $A(G)\setminus W\neq\varnothing$ it holds that $|W\cap \bigcup_{v\in G}A(v)|\geq \ell$.
\end{definition}
In words, a committee $W$ satisfies PJR+ if there is no group of voters who (i) deserve $\ell$ representatives, (ii) can all agree on a candidate not in $W$, but (iii) collectively support fewer than $\ell$ members of $W$. 


\begin{definition}{(The core)}
\label{def:core}
    A committee $W$ satisfies {\em core stability} \citep{Aziz2017JustifiedVoting,Lackner2023Multi-WinnerPreferences}, or is said to be \emph{in the core}, if $|W|=k$ and for every $\ell \in [k]$, every $T\subseteq C$ with $|T|\leq \ell$, and every group $G\subseteq V$ with $|G|\geq n\ell /k$ there is a voter $v\in G$ with $|W\cap A(v)|\geq |T\cap A(v)|$.
\end{definition}
Core stability is a more demanding property than PJR+: if a committee fails PJR+ then it is not in the core, but the converse is not true\footnote{This is claimed by Kalayci et al.~\citep{Kalayci2025FullRepresentation} without proof; for completeness, we provide a proof in \Cref{sec:corepjr+}.}.
Every MWV election has a committee that satisfies PJR+. In particular, the Method of Equal Shares (MES) \citep{Peters2020ProportionalityWelfarism} runs in polynomial time and outputs committees that satisfy PJR+\footnote{MES satisfies a stronger notion of proportionality called EJR+\cite{Brill2023RobustVoting}, which we shall discuss in \Cref{sec:conclusion}}. 
Later, we will use MES as a backup method to find a PJR+ committee (\Cref{alg:PJR}). 
On the other hand, it remains open whether every approval-based MWV election has a core stable committee \citep{Aziz2017JustifiedVoting,Lackner2023Multi-WinnerPreferences}.

We will consider one-dimensional approval preferences.
\begin{definition}{(Candidate-Interval (CI) preferences)}
    Given a set of candidates $C$ and a linear order $\lhd$ over $C$, we say that a subset of candidates $T\subseteq C$ is \emph{consecutive} if for every $x,z\in T$ and $y\in C$ with $x\lhd y\lhd z$ it holds that $y\in T$.
    An election $E=(V, C, k, A)$ has \emph{Candidate Interval (CI) preferences} if there exists a linear order $\lhd$ over $C$ such that for each voter $v$, the set $A(v)$ is consecutive \citep{Elkind2015StructurePreferences}.
\end{definition}
If $C\subseteq \mathbb{R}$, each voter approves of an interval of candidates: $T\subseteq C$ is consecutive if there exists an interval $I=[a,b]$ such that $T=C\cap I$. A visual example of CI preferences can be found in \Cref{fig:querying-example}.
One can interpret CI preferences as single-peaked preferences in the dichotomous setting; this class of preferences
has been explored in a number of papers, e.g,  \citep{Peters2020PreferencesCircle,Godziszewski2021AnElections,Terzopoulou2021RestrictedInformation,Pierczynski2022Core-StableDomains}.


\section{Random Interval Voter Model (RIV)}
\label{sec:ciriv}
In this section, we  introduce the probabilistic model for CI preferences that will be used throughout this paper; we will also briefly discuss an extension of this model to more than one dimension.

\begin{definition}
\label{def:sriv}
An {\em RIV model} is a $\sigma$-length list of 3-tuples $\mathcal{M} = (I_t, F_t, p_t)_{t\in[\sigma]}$, where $\sigma\in \mathbb{N}$, and for
each $t\in[\sigma]$ it holds that
$I_t=(z^-_t,z^+_t)$ is a subset of $\mathbb R$, 
with $I_x\cap I_y=\varnothing$ for $x\neq y$, $F_t$ is an invertible cumulative distribution function over $I_t$, $p_t\in[0, 1]$, and  $\sum_{t\in[\sigma]}p_t=1$. 
Given a candidate set  $C\subset\bigcup_{t\in[\sigma]}I_t$, we say that a voter $v$ is {\em sampled according to $\mathcal M$} if her ballot is obtained as follows:
a segment $I_t$ is chosen with probability $p_t$, two positions $X,Y\in I_t$ are drawn from distribution $F_t$, we set $a_v=\min(X, Y)$, $b_v=\max(X, Y)$,  
and the voter's approval ballot is defined as $A^*(v)=[a_v, b_v]$,  $A(v)=A^*(v)\cap C$.

For each $t\in[\sigma]$ we write $C_t=C\cap I_t$, 
and for each $x\in \bigcup_{t\in[\sigma]}I_t$ we define 
$t(x)$ to be the value $t\in[\sigma]$ with $x\in I_t$.

An RIV model $\mathcal M$ is {\em uniform} if $I_t=[t,t+1]$ and $F_t$ is the uniform distribution over $[t,t+1]$ for each $t\in[\sigma]$. Thus, if $\mathcal M$ is a uniform RIV model, we omit $(I_t, F_t)_{t\in[\sigma]}$ from the notation, and write ${\mathcal M} = (p_t)_{t\in[\sigma]}$.
\end{definition}

Note that focusing on uniform RIVs is
without loss of generality: If we are given a general RIV $\mathcal M$ together with a candidate set $C\subseteq \cup_{t\in[\sigma]}I_t$, we can transform it into a uniform RIV ${\mathcal M}'$ with a candidate set $C'\subseteq [1, t+1]$, by mapping each point (e.g., locations of candidates, approval endpoints, query positions) $x\in I_t$ in ${\mathcal M}$ to $t+F_t(x)$ in ${\mathcal M}'$.
Then a candidate is approved by a voter in ${\mathcal M}$ if and only if the transformed candidate is approved by the transformed voter in ${\mathcal M}'$, and it can be verified (see \Cref{sec:unif}) that this occurs with the same probability as in ${\mathcal M}$. 
Thus, from now on we will assume that all RIV instances are uniform.

We now define some notation that will be used throughout, and state a useful lemma.
\begin{definition}
    For each $T\subseteq C$ and $x\in \bigcup_{t=1}^\sigma I_t$, let 
    \begin{align*}
    \lefty{x}{T} & =\max(\{c\in T:c\leq x\}\cup\{z^-_{t(x)}\}), \\
    \righty{x}{T} &=\min (\{c\in T:c\geq x\}\cup\{z^+_{t(x)}\}).
    \end{align*}
\end{definition}
\noindent In words, $\lefty{x}{T}$ (resp., $\righty{x}{T}$) is the closest candidate in $T$ to the left (resp., to the right) of $x$, or the left (resp., right) endpoint of the segment $I_{t(x)}$ if no such candidate exists. 

The following lemma gives a useful expression for the probability of approving and disapproving candidates.
\begin{lemma}
\label{lem:appprob}
For a (uniform) RIV instance, a subset of candidates $S\subseteq C$, and a candidate $c\in C\setminus S$ we have
$$
    \Pr[c\in A(v)\land A(v)\cap S=\varnothing]=
    2p_{t(c)}\left(c - \lefty{c}{S}\right)\left( \righty{c}{S}-c\right).
$$
\end{lemma}
\begin{proof}
    A necessary condition for $v$ to approve $c$ is that both endpoints of $v$'s interval should be drawn from $I_{t(c)}$; moreover, one of them should come from $[z_{t(c)}^-, c]$, while the other should come from $[c, z_{t(c)}^+]$. 
    Now, suppose that this condition is satisfied, and
    consider the points $\lefty{c}{S}$ and 
    $\righty{c}{S}$; recall that, by definition, $\lefty{c}{S}, \righty{c}{S}\in I_{t(c)}$. For $A(v)\cap S=\varnothing$, it has to be the case  
    that the left endpoint of $v$'s interval is drawn from $[\lefty{c}{S}, c]$, while the right endpoint of her interval is drawn from $[c, \righty{c}{S}]$.
    
    Accordingly, the event that both $c \in A(v)$ and $A(v)\cap S=\varnothing$, i.e., $v$ approves $c$ but disapproves each $c'\in S$,  
    is identical to the event ${\lefty{c}{S}<a_v\leq c\leq b_v < \righty{c}{S}}$, 
    and we have 
    $$
        \Pr\left[c\in A(v)\land A(v)\cap S=\varnothing\right]=\\2p_{t(c)}\cdot (F_{t(c)}(c)-F_{t(c)}(\lefty{c}{S}))(F_{t(c)}(\righty{c}{S})-F_{t(c)}(c)).
    $$
    As we assume that the RIV instance is uniform, the result follows.
\end{proof}

\subsection{Queries}
We allow two types of queries: point queries and interval queries. A {\em point query} \textsc{Point}$(x,v)$ asks voter $v$ whether they approve candidate position $x$; the answer is $\mathds{1}[x\in A(v)]$. An {\em interval query} \textsc{Int}$(x, y, v)$ communicates two candidate or segment endpoint positions $x,y$ to $v$, who responds with $\mathds{1}[A^*(v)\subseteq [x,y]]$\footnote{\point$(x)$ is equivalent to $\neg\interval(0,x)\land\neg\interval(x,1)$; in other words, point queries are redundant with respect to interval queries, and are included only as syntactic sugar.}.
We refer to a sequence of queries of both types, together with voters' responses, as a {\em dialogue}; e.g., a dialogue can take the form 
${\mathcal I} = (c_1\in A(v), c_2\not\in A(v'), A^*(v'')\not\subseteq[c_3, c_4])$.

{\sc Point} queries allow us to obtain information about the approval of a single candidate, while {\sc Interval} queries allow us to effectively ask a voter ``Would you only approve of some compromise between $x$ and $y$?''
Note that we only submit {\sc point} queries at candidate positions: we may only ask the voters about actual candidates, and not about arbitrary positions on the real line. With interval queries, we additionally allow queries parametrized by the endpoints of segments. 

We note that under CI preferences each voter's ballot is fully determined by her leftmost and rightmost approved candidates; thus, there are at most 
$m^2+1$ distinct ballots. Consequently, if we consider a richer elicitation framework in which each voter can directly communicate which of the possible $m^2+1$ ballots she holds, we could elicit full information regarding the voters' preferences using $\log \left(m^2+1\right)=\mathcal{O}(\log m)$ queries per voter. In what follows, we show that we can elicit full preferences with $\mathcal{O}(\log m)$ queries in our model as well (\Cref{sec:obt}); moreover, 
it turns out that, with a much smaller number of queries (which, crucially, is independent of $m$), we can usually find a PJR+ outcome (\Cref{sec:pjr}). 


\section{Eliciting Full Preferences}
\label{sec:obt}
By asking a voter about every candidate in $C$, we can elicit her full preferences using $\Oh(m)$ {\point} queries. We will now show that,
in the 1-dimensional RIV model we can reduce this per-voter upper bound to $\Oh(\log m)$ in expectation. Importantly, our algorithm uses both {\point} and {\interval} queries.

We formalize the problem of eliciting a voter's preferences regarding a candidate subset $P\subseteq C$ as follows. 

\begin{definition}
We say that a voter $v$ is {\em resolved for $P\subseteq C$} given a dialogue
$\mathcal{I}$ if\,\, $\Pr[p\in A(v)\mid \mathcal{I}]\in \{0,1\}$ for all $p\in P$. We drop $\mathcal{I}$ from the notation when it is clear from the context.
\end{definition}
Our algorithm for resolving a given voter $v$ (\Cref{alg:perfectinfogeneric})
proceeds in two stages: the first stage is used to determine the segment $I_t$
containing $v$ (\Cref{alg:binsearch}), whereas the second stage is used to find $v$'s approval interval within $I_t$.

The first stage uses binary search, and is implemented via interval queries. Say $T$ is a union of segments. We
partition $T$ into two (nearly-)equal sized intervals $L,U$, where $L$ is the first $\floor{|T|/2}$ segments and $U=T\setminus L$,
and perform an interval query with $x,y=\min,\max L$ resp. We set $T:=L$ if the query response is positive and $T:=U$ if it is negative, and recurse. We stop when $T$ is a single segment; at that point, we check if $v$ lies in this unique interval $I_t$. If so, we output $t$, and if not, we know that $v$ did not lie in any of the original $T$ segments. As we essentially halve $T$ at each step, this stage is guaranteed to terminates after $\Oh(\log |T|)=\Oh(\log (\min(|P|,\sigma )))$ queries.

During the second stage, 
given $t=t(v)$, the algorithm
tries to find some candidate $c\in I_t\cap P$ with $c\in A(v)$, and then it runs two binary searches to determine the endpoints of $A(v)$.

To find some such $c$, it performs $\Oh(\log |P|)$ rounds of queries with exponentially increasing ``resolution''. 
Specifically, in round $i$, it submits a {\point} query on the candidates to the left and to the right of the position $t+j/2^i$ for $j\in[2^i]$. When some such query is answered positively, it yields a position in the voter's approval interval. The algorithm then performs two binary searches on either side of this position to find the candidates $\min(P\cap A(v))$ and $\max (P\cap A(v))$. Note that, since the algorithm did not terminate in the previous round (i.e., for $j-1$), it suffices to limit these binary searches to $\left[P\cap \left(t+\frac{j-1}{2^i},t+\frac{j}{2^i}\right)\right]$ and
$\left[P\cap \left(t+\frac{j}{2^i},t+\frac{j+1}{2^i}\right)\right]$.

After $\log |P|$ rounds, if a voter does not approve of any {\point} query it receives, the algorithm ``gives up'' and submits {\point} queries for all candidates in $P$ (the final two lines of \Cref{alg:perfectinfogeneric}). However, this only happens with low probability, so with high probability---and also in expectation---the algorithm uses $\Oh(\log |P|)$ queries.
\begin{algorithm}[t]
\SetKwFunction{FBinS}{SegmentSearch}
\SetKwProg{Fn}{Function}{:}{}
\Fn{\FBinS{$v$, $T$}}{
Sort $T$\;
lower$\gets0$; upper$\gets |T|-1$\;
\While{$\text{\em upper}>\text{\em lower}$}{
    
    mid$\gets \floor{(\text{upper}+\text{lower})/2}$\;
    \uIf{$\interval(T[\text{\em lower}],T[\text{\em mid}]+1, v)$ is approved}{
        upper$\gets$mid\;
    } \uElse{
        lower$\gets$mid+1\;
    }
}
\If{$\interval(T[\text{\em lower}],T[\text{\em lower}]+1)$ is approved}{\KwRet{lower}}
\Else{\KwRet{$\bot$}}

}
\caption{Finding which segment a voter appears in. $T$ is a possibly proper subset of $[\sigma]$, so $v$ may not be found in any segment of $T$.}
\label{alg:binsearch}
\end{algorithm}
\begin{algorithm}[t]
\SetKwFunction{FQuery}{resolve}
\SetKwFunction{FBinS}{SegmentSearch}
\SetKwProg{Fn}{Function}{:}{}
\Fn{\FQuery{$v$, $C$, $P$}}{
$t\gets \FBinS(v,\{t\in[\sigma]:P\cap I_t\neq \varnothing\})$\;
\lIf{$t$ is not found}{\KwRet{$\varnothing$}}
\For{$i\gets 1$ to $\ceil{\log (|P|)}$}{
    \For{$j\gets 1$ to $2^i$}{
    \textsc{Point}$\left(\lefty{\left(t+\frac{j}{2^i}\right)}{C}, v\right)$\; \textsc{Point}$\left(\righty{\left(t+\frac{j}{2^i}\right)}{C}, v\right)$\;
    \If{$v$ approves either point query}{
        Perform binary search point queries on
        $P\cap (t+\frac{j}{2^i},t+\frac{j+1}{2^i})$ and $P\cap (t+\frac{j-1}{2^i},t+\frac{j}{2^i})$ to find $\alpha = \min(P\cap A(v)), \beta=\max (P\cap A(v))$\;
        \Return $P\cap [\alpha, \beta]$\;
    }
    \textsc{Int}$\left(\left(\lefty{\left(t+\frac{j}{2^i}\right)}{C},\righty{\left(t+\frac{j}{2^i}\right)}{C}\right), v\right)$\;
    \lIf{$v$ approves this query}{
    \Return $\varnothing$ 
    }
    }
    
}
\tcp{Only occurs if not already returned}
\lFor{$c \in P\cap I_t$}{
\textsc{Point}$(c, v)$
}
\Return all approved points in $P$\;
}
\caption{Resolving a voter for a candidate set $P$}
\label{alg:perfectinfogeneric}
\end{algorithm}
\begin{example}
	We now describe how our algorithm resolves a single voter $v$. Suppose we have 8 segments, and $v$'s interval is $(7.1,7.35)\subseteq I_7$. The candidate positions are illustrated in \Cref{fig:querying-example}.
    \begin{figure}[t]
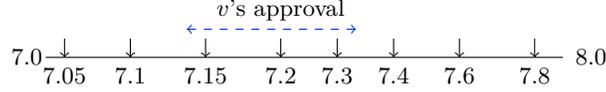

        \centering
        \ctikzfig{querying_example_nonscale}
        \caption{An example positioning of candidates (black arrows) and a voter's approval (blue dotted arrow) over the segment (not to scale)}
        \label{fig:querying-example}
    \end{figure}
First, we use interval queries to find which segment $I_t$ the voter belongs to. The first query is $\interval (1,5)$, and $v$ responds with ``false''. This rules out intervals $I_1,\dots,I_4$.  The algorithm then queries $\interval (5, 7)$ (the response is ``false'') and determines that $t=7$.
	Suppose now that $P\cap[7, 8]$ consists of the candidates shown by black arrows in \Cref{fig:querying-example}. During the second stage, we start by performing {\point} queries at $(7+1/2)^\pm_C$, i.e., $7.4$ and $7.6$; in both cases, the answer is ``false''. The next query is {sc Int}$(7.4, 7.6, v)$; again, the answer is ``false''. The next $\point$ query is at $(7+1/4)^\pm_C$, i.e., at $7.2$ and $7.3$. The voter responds ``true'' to both, and so we switch to performing two binary searches. We know that $v$ does not approve of $7.4$, and so we do not need to perform a binary search to the right of $7.3$. We perform a binary search to the left of $7.2$. The first {\point} query is $7.1$; the voter responds ``false''. The next query is $\point(7.15, v)$; the response is ``true''. We now have full information about $v$'s approval over $P$, and so we output $\{7.15,7.2,7.3\}$.
\end{example}
\begin{theorem}
\label{thm:querying}
    For every set $P\subseteq C$, given a uniform RIV, \Cref{alg:perfectinfogeneric} resolves a voter $v\in V$ using $\Oh(\log(|P|)$ queries in expectation.
\end{theorem}
    A key observation in our proof is that the probability that a voter does not approve of any {\sc Point} query in the first $i$ rounds is $1/2^i$.
    Hence, the probability that the algorithm ends up querying all points in $P$ can be bounded as $1/2^{\log |P|}=1/|P|$, 
    i.e., the expected cost of ``failure'' is $\Oh(1)$.
    Moreover, once a point $c\in A(v)$ has been identified, the two subsequent binary searches cost $O(\log |T|)$.
    We now present a full proof of \Cref{thm:querying}.
\begin{proof}
    Consider a voter $v$, and some set $P\subseteq C$. Let $\ell=\ceil{\log |P|}$. The initial binary search stage to identify $t(v)$ takes $\Oh(\log |P|)=\Oh(\ell)$ queries. Thus, if $t(v)$ is not found, the algorithm terminates after $\Oh(\ell)$ queries. 
    Therefore, from now on we assume that the algorithm successfully identified a value $t(v)\in[\sigma]$.
    
    Let $\rho$ be the random variable which takes value $i$ if \Cref{alg:perfectinfogeneric} terminates in round $i$ ($\rho=\ell+1$ if the algorithm has to query all of $P$). Let $\rho'$ be the smallest value of $i\in[\ell]$ such that $A^*(v)\cap\left\{t+j/2^{i}:j\in[2^{i}]\right\}\neq \varnothing$ (where $\rho'=\ell+1$ if this intersection is empty for all values of $i\in [\ell]$), so $\rho'$ is also a random variable. 
    Consider some $x\in A^*(v)\cap\left\{t+j/2^{\rho'}:j\in[2^{\rho'}]\right\}$. We know that $\rho\leq \rho'$; in round $\rho'$ the algorithm would perform $\point (\lefty{x}{C}, v)$ and $\point(\righty{x}{C}, v)$, and $\interval \left(\lefty{x}{C},\righty{x}{C}, v\right)$. 
    The voter would have to approve of at least one of these three queries; since $x\in A^*(v)$, we have $a_v\leq x\leq b_v$, and so if the respective interval query was not approved, $v$ must approve one of the point queries. 
    Of course, the algorithm may have terminated before reaching round $\rho'$, in which case $\rho<\rho'$, but it is certainly true that $\rho\leq \rho'$. 
    Let $Z$ be the event $\rho=\ell+1$ and $Z'$ the event that $\rho'=\ell+1$.
    Now, by the same argument that showed $\rho\leq \rho'$, we know that $\Pr[Z]\leq \Pr[Z']$; if none of the queries in the first $\ell$ rounds were approved, then it cannot be the case that any of the points $\left\{t+j/2^{\ell}:j\in[2^\ell]\right\}$ would be approved.
    We have that $$\Pr[Z']=\sum_{j\in[2^\ell]}\Pr\left[t+(j-1)/2^\ell<a_v\leq b_v<t+j/2^\ell\right]=\sum_{j\in[2^\ell]}2^{-2\ell}=2^{-\ell}\leq 1 /|P|.$$

    Let $Q$ be the total number of queries used by the algorithm. Note that if $\rho\leq \ell$, then $Q\leq \Oh(\ell)+ 3\times  2^{\rho}+2\log |P|\leq 3\times 2^{\rho'}+\Oh(\ell)$, and if $\rho=\ell +1$, then $Q\leq 3\times 2^{\ell}+|P|\leq 3\times  2^{\rho'}+|P|$. We have
    $
        E(2^{\rho'})=\sum _{r=1} ^{\ell}2^r {\Pr[\rho'=r]} + 2^{\ell+1} \Pr[Z']\leq \ell+ 4
    $
    where $\Pr[{\rho'=r}]=\Pr[\rho' = r\mid \rho'>r-1]\Pr[{\rho'>r-1}]=\frac{1}{2}\frac{1}{2^{r-1}}$ and hence $\Pr[\rho'=r]=2^{-r}$. Therefore,
    $$
        E(Q)\leq \Oh(\ell)+ 3 E\left(2^{\rho'} \right) +E\left(\begin{cases}
            |P| & \text{if }\rho=\ell+1\\
            \Oh(\ell) & \text{otherwise}
        \end{cases}\right)\\
        \leq \Oh(\ell)+\Oh(\ell)+|P|\Pr(Z')=\Oh( \log(|P|),
    $$
and hence the result is proved.
\end{proof}

\noindent 
By executing \Cref{alg:perfectinfogeneric} with $P=C$, we can learn all voters' complete preferences
using $\Oh(\log m)$ queries per voter. This enables us to find proportional committees. In particular, since for CI preferences the core is always non-empty, and a committee in the core can be computed in polynomial time given the voters' preferences \cite{Pierczynski2022Core-StableDomains}, we obtain the following corollary.

\begin{corollary}
\label{cor:core}
    Let $E$ be an election with candidate set $C$, $|C|=m$, obtained by sampling the voters' preferences from a uniform RIV. One can compute a  committee in the core of $E$ using $\mathcal{O}(\log m)$ queries per voter in expectation.
\end{corollary}

\section{Finding PJR+ Committees for RIV}\label{sec:pjr}
We have seen that if voters' preferences are sampled from a uniform RIV model, one can find a core stable outcome using $\Oh(\log m)$ queries per voter. We will now present our main result: for a slightly less ambitious notion of proportionality, namely, PJR+, one can find proportional outcomes using $\Oh(\log(\sigma k))$ queries per voter in expectation; notably, 
the expected number of queries does not depend on $m$.
In particular, this means that our algorithm for PJR+ does not elicit voters' full preferences; in that sense, it follows the line of work on learning solution concepts in games~\cite{BalcanPZ15,JhaZ23}.

\begin{theorem}
\label{thm:PJR}
   For a uniform RIV, \Cref{alg:PJR} finds a PJR+ committee. When the number of voters $n$ satisfies $n=\Omega(k^2\log m)$, the algorithm uses $\mathcal{O}( \log (\sigma k))$ queries per voter in expectation and with high probability.
\end{theorem}

The proof of \Cref{thm:PJR} is based on a sequence of algorithms (\Cref{alg:PJR,alg:validate,alg:possible}).
Briefly, our algorithm first ``guesses'' a committee $\W$; we will argue that $\W$ provides PJR+ with high probability. It then elicits preference information from each voter to obtain an approximation to the voter's true approval window.
Given this information, the algorithm then checks whether $\W$ necessarily provides PJR+. If the check fails (which happens with low probability), the algorithm elicits voters' full preferences, and computes a PJR+ committee using the Method of Equal Shares.   

Our guessing stage proceeds as follows. For each $t\in[\sigma]$, we 
allocate $k_t=\floor{p_t k}$ committee seats to $I_t$. Specifically, we mark $k_t$ points
$\left\{t+\frac{2\xi-1}{2(k_t+2)}:\xi\in[k_t+1]\setminus\{1\}\right\}$. We then
form the set $\W_t$ by selecting $k_t$ candidates from $C_t=C\cap I_t$ that are closest to the marked points. We then set $\W=\bigcup_{t=1}^\sigma\W_t$. Intuitively, in this way we ensure that the selected candidates are approximately uniformly distributed on each segment.
    \begin{figure}[t]
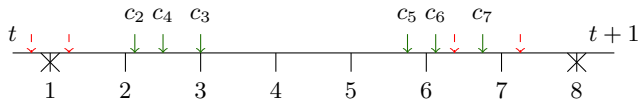

    \ctikzfig{algorithm}
    \caption{The algorithm selects candidates $c$ for the committee $\W_t$, in order $c_2,\dots,c_7$. Unelected candidates are indicated by dotted red arrows.}
    \label{fig:example}
\end{figure}
\begin{example}
We provide an example of how we construct $\W$. Let $\sigma=1$ (an extension to multiple segments is straightforward), and consider candidate positions in \Cref{fig:example} with $k=k_1=6$.
We iterate over marked points from left to right, selecting the candidate that is closest to the marked point.
Among all candidates, $c_2$ has the minimum distance to marked point 2, so we select it. We then pick $c_3$ for marked point 3. For marked point 4, note that $c_3$ was already selected, so the closest unselected candidate is $c_4$. Eventually, we obtain $\W=\{c_2,\dots,c_7\}$, $|\W|=k=6$. We see that the algorithm spreads $k$ candidates across the segment as uniformly as possible.
\end{example}
\begin{algorithm}[p]
    \SetNoFillComment
    \SetKwFunction{FQuery}{resolve}
    \SetKwFunction{FValidate}{validate}
    \SetKwFunction{FPossible}{poss}
    \SetKwFunction{FBinS}{SegmentSearch}
    \SetKwProg{Fn}{Function}{:}{}
    \SetKwInOut{Input}{input}
    \SetKwInOut{Output}{output}
    \Input{A uniform RIV $(p_t)_{t\in [\sigma]}$}
    
    \For{$t\in[\sigma]$}{
    $k_t\gets \floor{p_tk}$\;
    $\W_t\gets \varnothing$\;
    \For{$\xi=2$ to $k_t+1$}{
        $c\gets \arg\min_{c\in C_t\setminus \W_t} \big|c-t-\frac{2\xi-1}{2(k_t+2)}\big|$\;
        $\W_t\gets \W_t\cup \{c\}$\;
    }
    }
    $\W=\bigcup_{t=1}^\sigma \W_t$, sorted\;
    \For{$t\in [\sigma]$}{
        $P_t\gets \W_t\cup\left\{\adj{\left(t+\frac{j}{15 k }\right)}{C_t}{+/-}:j\in[15 k ]\right\}$\;
    }
    \For{$v\in V$}{
        $t\gets \FBinS(v,[\sigma])$\;
        $\widehat{A}(v)\gets \FQuery(v, C, P_t)$\;
    }
    $\widehat{E}\gets\left(V,\bigcup_{t\in [\sigma]}P_t,k,\widehat{A}\right)$\;
    \If{\FValidate($\W,\widehat{E},C$) (Algorithm \ref{alg:validate})}{
        \Output{$\W$}
    }
    \For{$v\in V$}{ 
        $A(v)\gets\FQuery(v, C, C)$\;
    }
    {$E\gets(V,C,k,A)$\;}
    \Output{$W^*=\texttt{MES}(E)$}
    
    \caption{Finding a PJR+ committee}
    \label{alg:PJR}
    \end{algorithm}
    \begin{algorithm}[p]
    \SetNoFillComment
    \SetKwFunction{FValidate}{validate}
    \SetKwFunction{FPossible}{poss}
    \SetKwFunction{FRank}{$\rho$}
    \SetKwProg{Fn}{Function}{:}{}
    \SetKwInOut{Input}{input}
    \SetKwInOut{Output}{output}
    
    \Fn{\FValidate{$\W$, $\widehat{E}$, $C$}}{
    $u\gets|\{v\in V:v\text{ has not approved any query}\}|$\;
    \For{$\ell=1$ to $k$}{
        \For{$c\in C\setminus\W$}{
        $\rho\gets |\{x\in \W:x<c\}|$\;
        \For{$j=0$ to $\ell-1$ }{
            $s_{\ell,c,j}\gets 0$\;
            $R\gets \W[\rho-j:\rho+\ell-1-j]$\;
            \For{$v\in V$ that approved some query}{
                \If{\FPossible$(v,c,\W\setminus R)$ (Algorithm \ref{alg:possible})}{
                    $s_{\ell,c,j}\gets s_{\ell,c,j}+1$\;
                    \If{$s_{\ell,c,j}+u\geq \frac{n\ell}{k}$}{
                        \KwRet{\emph{False}}\;
                    }
                }
            }
        }
        }
    }
    \KwRet{\emph{True}}\;
    }
    \caption{Validating if a committee $\W$ provides PJR+}
    \label{alg:validate}
    \end{algorithm}
    \begin{algorithm}[p]
    \SetNoFillComment
    \SetKwFunction{FPossible}{poss}
    \SetKwProg{Fn}{Function}{:}{}
    \SetKwInOut{Input}{input}
    \SetKwInOut{Output}{output}
    \Fn{\FPossible{$v$, $c$, $S$}}{
    \If{$c\in S$}{\KwRet{False}}
    $A\gets$ set of {\sc Point} query positions $v$ has approved\;
    $D\gets$ set of {\sc Point} query positions $v$ has disapproved\;
    $\phi_2\gets \min A$\;
    $\phi_3\gets\max A$\;
    $\phi_1\gets \lefty{\left(\phi_2\right)}{D}$\;
    $\phi_4\gets \righty{\left(\phi_3\right)}{D}$\;
    \KwRet{$(\phi_1,\phi_2]\cap(\lefty{c}{S},c]\neq \varnothing\land[\phi_3,\phi_4)\cap[c,\righty{c}{S})\neq \varnothing$}\;        
    }
    \caption{Determining if $v$ has positive probability of approving $c$ and disapproving $S$.}
    \label{alg:possible}
    \end{algorithm}
Given our guessed committee $\W$, we verify whether it provides PJR+. For each voter $v\in V$, we use {\tt SegmentSearch} to find the segment $I_t$ that contains $v$. We then pick points spaced approximately $1/15k$ apart along the segment, along with $\W$, and use \Cref{alg:perfectinfogeneric} to resolve each voter on the candidates closest to these points. This gives the algorithm an approximation for the voter's approval interval.
If the voter did not approve any query, we add them to the counter $u$, as we do not have much information about the location of the voter's approval interval, except that it must lie entirely between two adjacent {\sc Point} queries.

If the voter did approve some query, we check for every $\ell\in [k]$, every candidate $c\in C\setminus \W$ and for every $0\leq j < \ell$ whether it could be consistent, given the information we have elicited from the voter, for the voter to approve of: (1) $c$, (2) at most $j$ candidates in $\W$ to the left of $c$, and (3) fewer than $\ell$ candidates in $\W$.
If this event has positive probability, captured by the \texttt{poss} function, we add the voter to the counter $s_{\ell,c,j}$. Finally, if $s_{\ell,c,j}+u< n\ell/k$ for all $\ell\in[k]$, $c\in C\setminus\W$, $0\le j<\ell$, we output $\W$. Otherwise, we query all voters in $V$ about all candidates in $C$ to get full information about $E$, and then run MES to output a committee $W^*$ that provides PJR+.

In \Cref{sec:WPJR} we show that if $s_{\ell,c,j}+u< n\ell/k$ for all $\ell\in[k]$, $c\in C\setminus\W$, $0\le j<\ell$, then $\W$ is guaranteed to provide PJR+. 
Otherwise, we use $\mathcal{O}(m)$ queries per voter in order to obtain full information about $A(v)$ for each $v\in V$. However, we show that the latter event is rare, so with high probability and in expectation we require $\mathcal{O}(\log \sigma k)$ queries.
The time complexity of \Cref{alg:PJR} (without MES, which runs in $\Oh(m^2n\log n)$ \citep{Peters2021ProportionalUtilities}) is $\Oh(nmk^2)$.

To prove \Cref{thm:PJR}, we require additional notation; we first discuss it informally, and then provide rigorous definitions.

The quantities $\W,\W_t,s_{\ell,c,j}$ are defined in Algorithms~ \ref{alg:PJR} and~\ref{alg:validate}
(in what follows, we consider the values of these variables at the termination of the respective algorithms).

Assume that all $\W_t$ are ordered from left to right. 
The function $\rho(c)$ counts the candidates in $\W_{t(c)}$ to the left of $c$, so that $\W_{t(c)}[\rho(c) -r]$ is the $r$-th candidate in $\W_{t(c)}$ to the left of $c$. 
The set $K_t$ contains the candidates in the subinterval of $I_t$ bounded by the first and last marked points (e.g., marked points 2 and 7 in \Cref{fig:example}). 
For each $c\in K_t$, there exists at least one marked point to the left and to the right of $c$, so we can prove probability bounds on $c$; 
we deal with candidates in $C_t\setminus K_t$ separately. 
We use $d^{+/- }_r (c)$ to denote the distance between $c$ and the $(r+1)$-th candidate in $\W_{t(c)}$ to the right/left of $c$, respectively (if there are fewer than $r+1$ candidates to the left/right of $c$ in $\W_{t(c)}$, we measure the distance to the respective endpoint of $I_{t(c)}$ instead); see \Cref{fig:distanceexample} for an illustration. 
Finally, $Y_{\ell,c,j}$ is the event that a voter approves candidate $c$, at most $j$ candidates in $\W$ to the left of $c$, and strictly fewer than $\ell$ candidates in $\W$ in total. 
Now we provide formal definitions.

\begin{definition}
\label{def:bigdef}
     Set 
     \begin{align*}
     \rho(c) & =|\{c'\in \W_{t(c)}:c'<c\}|, \\  K_t &=C_t\cap\left(t+\frac{3}{2(k_t+2)},t+1-\frac{3}{2(k_t+2)}\right).
     \end{align*}
      For $0\leq j <\ell\leq k$ and $c\in C\setminus \W$, define the random event 
      $Y_{\ell,c,j}$ as \[ c \in A(v)\text{ and }
        A(v)\cap \W \subseteq \W_{t(c)}[\rho(c)-j:\rho(c)+\ell-1-j].
        \]  
    Let $\pi(\ell,c,j)=\Pr[Y_{\ell,c,j}]$ and 
    $$
    d^+_{r}(c)=\W_{t(c)}[\rho(c)+r]-c, \quad 
    d^-_{r}(c)=c-\W_{t(c)}[\rho(c)-r-1].
    $$
\end{definition}

\begin{figure}[t]
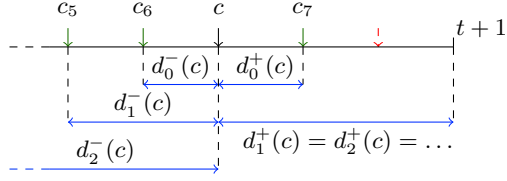


    \centering
    \ctikzfig{algorithmexample}
    \caption{The quantity $d^{+/-}_r (c)$ is the distance from candidate $c\in C\setminus \W$ to the $(r+1)$-th candidate in $\W$ to the left/right of $c$, or to the segment endpoints if no such $\W$ candidate exists.}
    \label{fig:distanceexample}
\end{figure}
We can express $\pi(\ell,c,j)=\Pr[Y_{\ell,c,j}]$ in terms of $d^{+/-} _r(c)$, and bound $d^+ _{r_1}(c)+d^-_{r_2}(c)$ to obtain a bound on $\pi$ in terms of $\ell$ and $k$.
\begin{lemma}
\label{lem:dlem}
    $\pi(r_1+r_2+1,c,r_1)=2p_{t(c)}\cdot d^- _{r_1}(c)\cdot d^+ _{r_2}(c).
    $
 \end{lemma}

\begin{proof}
 $\pi(r_1+r_2+1,c,r_1)$ is the probability that $v$ approves of $c$ and at most $r_1$ candidates in $\W$ to the left of $c$ and at most $r_2$ candidates to the right of $c$. We only need to consider $\W_{t(c)}$ (the candidates of $\W$ in the same segment as $c$),  not all of $\W$, since no voter can approve $c$ and some candidate in $\W\setminus\W_{t(c)}$.
 We have
$
A(v)\cap \W\subseteq \W_{t(c)}[\rho(c)-r_1:\rho(c)+r_2]$ if and only if $A(v) \cap (\W_{t(c)}\setminus \W_{t(c)}[\rho(c)-r_1:\rho(c)+r_2])=\varnothing.
$ Let $S=\W_{t(c)}\setminus \W_{t(c)}[\rho(c)-r_1:\rho(c)+r_2]$.
Then, with \Cref{lem:appprob}, we have
$$
    \Pr[c\in A(c)\land A(v) \cap S)=\varnothing]
=\\2p_{t(c)}(c-\lefty{c}{S})(\righty{c}{S}-c)=2p_{t(c)}\cdot d_{r_1}^-(c)\cdot d_{r_2}^+(c);
$$
note that if $\rho(c)-r_2<1$, then we say $\lefty{c}{S}=t(c)$, and similar for $\rho(c)+r_2\geq k_{t(c)}$, we have $\righty{c}{S}=t(c)+1$.
 \end{proof}

Given this result, we can transform a bound on $d^{+/-}_r$ into a bound on $\pi$. We first provide a bound on $d^{+} _{r_1} (c)+d^{-} _{r_2} (c)$.
\begin{lemma}
\label{lem:distlim}
    $
    d^+ _{r_1}(c)+d^- _{r_2}(c)\leq \frac{2(r_1+r_2+1)}{k_{t(c)}+2}
    $ for $c\in K_t\setminus \W_t$.
\end{lemma}
\begin{figure}[t]
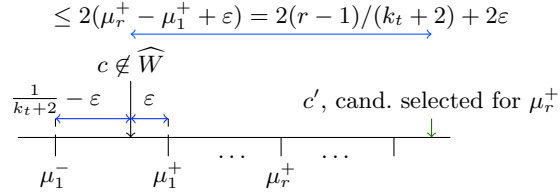


    \centering
    \ctikzfig{distintuiti}
    \caption{Intuition for \Cref{lem:distlim}. The candidate selected for a marked point $\mu^+_r$ cannot be further from $\mu^+_r$ than a candidate $c\not \in W$, as otherwise $c$ would have been chosen instead.}
    \label{fig:distlimint}
\end{figure}
The full proof of this result can be found in \Cref{sec:distlimproof}. The general intuition for this result can be seen in \Cref{fig:distlimint} and is as follows: recall that in each segment $I_t$, we select $k_t$ candidates, according to their distance to the marked points $t+\frac{2\xi -1}{2(k_t+1)}$. Let $c'$ be the candidate selected for the $r$-th marked point $\mu_r^+$ to the right of $c$. Given that $c\not\in \W_{t}$, we know that $|c'-\mu_r^+|\leq |c-\mu^+_r|$, as otherwise $c$ would have been chosen by $\mu ^+_r$ instead.
This allows us to upper bound $d_{r_1}^+(c)$ as $d_{r_1}^+(c)\leq\frac{2r_1}{k_t+2}+2\varepsilon$, where $\varepsilon=\mu^+_1-c$. We can derive a similar bound on $d_{r_2}^-(c)$, namely, $d_{r_2}^-(c)\leq \frac{2r_2}{k_t+2}+2\left(\frac{1}{k_t+2}-\varepsilon\right)$, giving us the desired result. 

With \Cref{lem:dlem}, we can bound $\pi(\ell,c,j)$, the probability that a voter approves $c$, at most $j$ candidates in $\W_t$ to the left of $c$, and fewer than $\ell$ in total for all $c\in C_t\setminus \W$. 
\begin{corollary}
\label{cor:prob}
    For all $c\in K_t\setminus \W$ we have
    $\pi(\ell,c,j)\leq \frac{2p_{t}\ell^2}{(k_t+2)^2}.$
\end{corollary}
\begin{lemma}
\label{lem:combinedbound}
    For all $c\in C_t\setminus \W$ we have
    $\pi(\ell,c,j)\leq \ell/k-1/4k.$
\end{lemma}
The proof of \Cref{lem:combinedbound} can be found in \Cref{sec:combinedboundproof}.
Now that we have upper bounds on the probability that a given voter may be a member of a group that has the potential to violate PJR+, we can bound the probability that $s_{\ell,c,j}\geq \ell n/k-n/(12k)$.
\begin{lemma}
\label{lem:sbound}
For a given $\ell,c,j$, we have
    $\Pr\left[s_{\ell,c,j}\geq \frac{\ell n}{k}-\frac{n}{12k}\right]\leq\exp\left(-\alpha\frac{n}{k^2}\right),$ where $\alpha=1800^{-1}$ (a constant)
\end{lemma}
Briefly, for a given voter $v$, while we do not directly observe whether the event $Y_{\ell,c,j}$ (see \Cref{def:bigdef}) has occurred, we perform queries sufficiently close to $c$ such that the probability that $\texttt{poss}(v,c,\W\setminus R)$ evaluates to true is within $3/(20k)$ of $\pi(\ell,c,j)$, the true probability of $Y_{\ell,c,j}$. Since $\pi(\ell,c,j)\leq \ell/k-1/(4k)$, we find that the probability that $\texttt{poss}(v,c,\W\setminus R)$ evaluates to true, and therefore that $v$ gets counted towards $s_{\ell,c,j}$ is at most $\pi(\ell,c,j) +3/(20k)\leq \ell/k-1/(10k)$. We then use a Hoeffding bound \citep{Hoeffding1963ProbabilityVariables} to show that $\Pr\left[s_{\ell,c,j}\geq \frac{\ell n}{k}-\frac{n}{12k}\right]\leq\exp\left(-\frac{\alpha n}{k^2}\right).$
\begin{proof}[Proof of Lemma \ref{lem:sbound}]
    Let us bound $s_{\ell,c,j}$ for $c\in I_t$. 
    Assume $c\not\in \W$. Let $R$ be as defined in the algorithm. 
    Let $X_v$ be the event that $\texttt{poss}(v,c,\W\setminus R)$ evaluates to true. This is the random event that $v$'s approval interval $[a_v,b_v]$ is such that, given the information $\mathcal{I}$ obtained from querying $v$, it holds that $\Pr[Y_{\ell,c,j}\mid \mathcal{I}]>0$.
    We want to find
    $\Pr[X_v]$.
    Let $S={\W\setminus R}$. We have $\delta^-\coloneqq d^-_{j}(c)=c-\lefty{c}{S}$, and $\delta^+\coloneqq d^+_{\ell-1-j}(c)=\righty{c}{S}-c$. We know that $\delta^-+\delta^+\leq 1$, and by \Cref{lem:appprob,lem:combinedbound} we have $
       2p_t\delta^+\delta^- =\pi(\ell,c,j)\leq \frac{\ell}{k}-\frac{1}{4k}
    $. First, if $c\in P$ (as defined in \Cref{alg:PJR}), then $\Pr[X_v]=\pi(\ell,c,j)$, since we will perform queries at $c$ and at all of $S$, and so $Y_{\ell,c,j}=X_v$. 
    
    Now, suppose $c\not \in P$; see \Cref{fig:querylims} for a diagram. Let $q_r=t+r/(15 k)$ be the $r$-th position in segment $t$ around which we query. 
    Then consider $r$ with $q_r\leq c < q_{r+1}$. 
    Since $c\not \in P$, we know that $q_r\leq \righty{\left(q_r\right)}{C}<c<\lefty{\left(q_{r+1}\right)}{C}\leq q_{r+1}$, and so there exist queried positions either side of $c$ with distance at most $1/(15 k)$ between them.
    Recall that we query $v$ with all candidates in $\W$, and so we have no uncertainty as to whether $v$ approves of $\lefty{c}{S}$ or $\righty{c}{S}$.
    Assuming that the voter has approved of some query and $a_v,b_v$ are the endpoints of their approval interval, then if $X_v$ occurs, then $a_v\in (\lefty{c}{S},\lefty{(q_{r+1})}{C}]$ and $b_v\in [\righty{(q_{r})}{C},\righty{c}{S})$. Indeed, if $a_v\leq\lefty{c}{S}$ or $a_v>\lefty{(q_{r+1})}{C}$, then given the elicited information, it would be impossible for both $c$ to be approved and $\lefty{c}{S}$ to be disapproved, and analogously for $b_v$.
    
    We then have
    $
        \Pr[X_v]\leq2p_t\left(\delta^++\frac{1}{15 k}\right)\left(\delta^-+\frac{1}{15k}\right)\leq\frac{\ell}{k}-\frac{1}{10k}.
    $\begin{figure}[t]
    \ctikzfig{querylimssmallcol}
    \caption{Intuition for \Cref{lem:sbound}. Query points are shown in green. Recall that we query at $\lefty{c}{S}$ and $\righty{c}{S}$, since $S\subseteq \W$.}
    \label{fig:querylims}
\end{figure}
    Now, $s_{\ell,c,j}\leq \sum_{v\in V}X_v$ is upper bounded by a Binomial random variable, and our result follows from a Hoeffding bound \citep{Hoeffding1963ProbabilityVariables}:
    $$
    \Pr\left[\text{Bin}\left(n,\frac{\ell}{k}-\frac{1}{10k}\right)\geq \frac{\ell n}{k}-\frac{n}{12k}\right]\leq \exp\left(-\alpha n/k^2\right).
    $$
\end{proof}

We can now prove \Cref{thm:PJR}.
\begin{proof}[Proof of \Cref{thm:PJR}]
    First, the output of \Cref{alg:PJR} always provides PJR+\footnote{Details can be found in \Cref{sec:WPJR}.}: $W^*$ (the output of MES on full information) provides PJR+, and if $\W$ is output by the algorithm, the voters in any large cohesive group would be counted by some $s_{\ell,c,j}$ or by $u$.

    Let us bound the probability that $u+s_{\ell,c,j}\geq n\ell/k$. To bound $u$, we want to find the probability that a voter does not approve of any {\sc Point} or {\sc Interval} query made during $\texttt{resolve}(v,C,P)$. Recall the proof of \Cref{thm:querying}, where $Z'$ is the event that no query to $v$ is approved; we have $\Pr[v\in U]\leq\Pr[Z']\leq 1 /|P|\leq1/15k$.
We have a standard Hoeffding bound \citep{Hoeffding1963ProbabilityVariables} with $\alpha=1800^{-1}$:
$
\Pr\left[u\geq \frac{n}{12k}\right]\leq \Pr\left[\text{Bin}\left(n,\frac{1}{15 k}\right)\geq \frac{n}{12k}\right]\leq \exp\left(-\frac{\alpha n}{k^2}\right).
$
    We can also bound $s_{\ell,c,j}$ with \Cref{lem:sbound}.
    By the union bound, the probability that $s_{\ell,c,j}\geq \frac{\ell n}{k}-\frac{1}{12k}$ for any $s_{\ell,c,j}$ is at most $(k^2m-1)\exp\left(-\alpha n/k^2\right)$. 
    Finally, combining that by union bound with the probability that $u\geq n/12k$, we conclude that $\Pr[u+s_{\ell,c,j}\geq n\ell/k]\le k^2m\exp\left(-\alpha n/k^2\right)$. 
    
    Hence, if $n=\Omega(k^2\log (m^2k^2))=\Omega(k^2\log m)$, we can bound the expected number of queries per voter as $E(Q)=\Oh(\log (\sigma k))$. Indeed, we use $\Oh(\log \sigma)$ queries to determine which segment $v$ lies in, and the expected number of queries is
    \begin{equation*}
        E(Q)\leq
        k^2m^2\exp\left(-\alpha n/k^2\right)
        +\Oh(\log \sigma+\log |P|)\leq\Oh(\log \sigma |P|).
    \end{equation*}
    With $|P|=\Oh(k)$, we have $E(Q)\leq \Oh(\log (\sigma k))$. Note that PJR+ still holds if $n=o(k^2\log m)$, but we may need more than $\Oh(\log \sigma k)$ queries per voter in expectation.
\end{proof}


\section{Conclusion and Future Work}
\label{sec:conclusion}
We introduced a new random voter model for multiwinner voting with one-dimensional approval preferences called the Random Interval Voter model (RIV), as well as a query framework for this model. Given an RIV election, we can find a committee that satisfies PJR+ using $\mathcal{O}(\log \sigma k)$ queries per voter in expectation. 

The query complexity, as defined in our work, can be seen as a crude heuristic for \emph{cognitive burden}.
We believe that $\mathcal{O}(\log \sigma  k)$ is a reasonable measure of what a voter can evaluate; however, the constant in the analysis may be too large. We note that we can reduce $|P|$ used in \Cref{alg:PJR} by a constant factor, at the expense of requiring $n$ to be larger by a constant factor.

An immediate open question is whether our approach can be extended to EJR+\,---a strictly stronger notion of fairness than PJR+, which requires each $\ell$-cohesive group to contain some voter $v$ with $|A(v)\cap W|\geq \ell$.
To adapt our analysis to EJR+, we would have to bound the sum $$\sum_{t=0}^{\ell-1}\left(d^-_{r_1}(c)-d^-_{r_1-1}(c)\right)d_{\ell-1-r_2}^+(c)\leq (\ell-\lambda)/2k_t$$ for some constant $\lambda>0$.
Using techniques similar to those in \Cref{sec:pjr}, we can obtain a bound on $d^-_{r_1}(c)-d^-_{r_1-1}(c)$ that is tight for each individual value of $r_1$, in that there exists an election where the bound matches for a given $r_1$; however, this bound is loose in that we cannot construct an election such that the bound is tight for all $r_1,r_2$, and it is too loose to prove EJR+. Indeed, the elections that are tight for a given $r_1$ are often very loose for other values of $r$, so we would need a more sophisticated bound that simultaneously considers multiple values of $r$.

Another important direction is to establish lower bounds on the amount of information required in this model to guarantee forms of justified representation. The fact that we focus on \emph{expected} rather than maximum number of queries makes it more challenging to establish lower bounds.

It would also be interesting to explore an alternative probabilistic model where each voter has an ideal position and a radius around that position they are willing to approve, with both the ideal position and the radius drawn from given distributions \citep{Bredereck2019AnRepresentation,Elkind2023JustifyingVoting,Szufa2022HowElections}. 

Finally, it is natural to ask if our results can be extended beyond one dimension. Indeed, spatial models with $d>1$ are popular in the literature \citep{Miller2018TheVoting,Godziszewski2021AnElections,Imber2024SpatialInformation}: each dimension corresponds to an issue, and a point in this space (a candidate) represents a specific stance on all of the issues simultaneously. \citet{Lewenberg2016PoliticalModel} suggest that political opinions in the UK lie in a 10-dimensional space, so expanding our model to higher dimensions may be necessary to accurately model real-world political elections. 

Towards this goal, we propose the $d$-dimensional Random Euclidean Voter Model, which is a generalization of RIV.
\begin{definition}
\label{def:generald}
Every bounded interval $I$ of $\mathbb{R}$ together with a finite set $C\subset I^d$, a probability distribution $\mathcal{D}$ over convex subsets of $I^d$, 
and $n\in\mathbb{N}$ define a 
{\em $d$-dimensional Random Euclidean Voter Model ($d$-REV)} as follows: we draw $n$ samples $A^*_1, \dots, A^*_n\sim \mathcal{D}$, 
let $V=[n]$, and define the approval set of voter $v\in V$
as $A(v)=A^*_v\cap C$. 
    We define two types of queries that can be used to elicit voters' preferences:
    \begin{itemize}
        \item Point queries: Given a candidate $x\in C$ and a voter $v\in V$, the query {\point}$(x, v)$ returns 
         $\mathds{1}[x\in A^*(v)]$.
        \item Hull queries: Given a finite set of points $P\subseteq I^d$ and a voter $v\in V$, the query {\sc Hull}$(P, v)$ computes $H$ as the convex hull of points
        in $P$ and returns $\mathds{1}[A^*(v)\subseteq H]$.
    \end{itemize}
\end{definition}
We note that point queries and hull queries in the $d$-REV model correspond exactly to point and interval queries in the RIV model.
In general, voters' approval sets may have complex geometries, so that communicating (a description of) a voter's approval set may require a significant amount of information. However, the two types of queries above enable efficient communication; each query consists of a finite number of points in $I^d$, and the voter's response is always a single bit.

The spatial model with $d>1$ is more expressive than the RIV model; in particular, it can capture important aspects of participatory budgeting. Indeed, consider a $2$-REV model representing the geographical layout of a city, where each project is associated with a particular location, and voters approve of all projects that lie within some area around their residence. We are able to query voters in two ways: (1) we can ask a voter whether they approve a specific project, and (2) we can ask a voter whether all of their approved projects lie within the convex hull of  some set of points on the map. We hope that our investigation of the 1-dimensional model may provide useful insights into how to tackle higher-dimensional settings.





\bibliographystyle{plainnat}
\bibliography{arxiv}

\begin{thebibliography}{37}
\providecommand{\natexlab}[1]{#1}
\providecommand{\url}[1]{\texttt{#1}}
\expandafter\ifx\csname urlstyle\endcsname\relax
  \providecommand{\doi}[1]{doi: #1}\else
  \providecommand{\doi}{doi: \begingroup \urlstyle{rm}\Url}\fi

\bibitem[Aziz and Lee(2021)]{Aziz2021ProportionallyPreferences}
Haris Aziz and Barton~E Lee.
\newblock {Proportionally representative participatory budgeting with ordinal
  preferences}.
\newblock In \emph{Proceedings of the 35th AAAI Conference on Artificial
  Intelligence, AAAI 2021}, pages 5110--5118, 2021.

\bibitem[Aziz et~al.(2017)Aziz, Brill, Conitzer, Elkind, Freeman, and
  Walsh]{Aziz2017JustifiedVoting}
Haris Aziz, Markus Brill, Vincent Conitzer, Edith Elkind, Rupert Freeman, and
  Toby Walsh.
\newblock {Justified representation in approval-based committee voting}.
\newblock \emph{Social Choice and Welfare}, 48\penalty0 (2):\penalty0 461--485,
  2017.

\bibitem[Balcan et~al.(2015)Balcan, Procaccia, and Zick]{BalcanPZ15}
Maria{-}Florina Balcan, Ariel~D. Procaccia, and Yair Zick.
\newblock Learning cooperative games.
\newblock In \emph{Proceedings of the 24th International Joint Conference on
  Artificial Intelligence, {IJCAI} 2015}, pages 475--481, 2015.

\bibitem[Bredereck et~al.(2019)Bredereck, Faliszewski, Kaczmarczyk, and
  Niedermeier]{Bredereck2019AnRepresentation}
Robert Bredereck, Piotr Faliszewski, Andrzej Kaczmarczyk, and Rolf Niedermeier.
\newblock An experimental view on committees providing justified
  representation.
\newblock In \emph{Proceedings of the 28th International Joint Conference on
  Artificial Intelligence, IJCAI 2019}, pages 109--115, 2019.

\bibitem[Brill and Peters(2023)]{Brill2023RobustVoting}
Markus Brill and Jannik Peters.
\newblock {Robust and verifiable proportionality axioms for multiwinner
  voting}.
\newblock In \emph{Proceedings of the 24th {ACM} Conference on Economics and
  Computation, {ACM EC} 2023}, page 301, 2023.

\bibitem[Conitzer and Sandholm(2002)]{Conitzer2002VoteStrategy-Proofness}
Vincent Conitzer and Tuomas Sandholm.
\newblock Vote elicitation: Complexity and strategy-proofness.
\newblock In \emph{Proceedings of the 18th National Conference on Artificial
  Intelligence, AAAI 2002}, pages 392--397, 2002.

\bibitem[Dery et~al.(2016)Dery, Kalech, Rokach, and
  Shapira]{Dery2016ReducingMaking}
Lihi~Naamani Dery, Meir Kalech, Lior Rokach, and Bracha Shapira.
\newblock {Reducing preference elicitation in group decision making}.
\newblock \emph{Expert Syst. Appl.}, 61:\penalty0 246--261, 2016.

\bibitem[Dey and Bhattacharyya(2015)]{Dey2015SampleElections}
Palash Dey and Arnab Bhattacharyya.
\newblock Sample complexity for winner prediction in elections.
\newblock In \emph{Proceedings of the 2015 International Conference on
  Autonomous Agents and Multiagent Systems, AAMAS 2015}, pages 1421--1430,
  2015.

\bibitem[Elkind and Lackner(2015)]{Elkind2015StructurePreferences}
Edith Elkind and Martin Lackner.
\newblock Structure in dichotomous preferences.
\newblock In \emph{Proceedings of the 24th International Joint Conference on
  Artificial Intelligence, IJCAI 2015}, pages 2019--2025, 2015.

\bibitem[Elkind et~al.(2023)Elkind, Faliszewski, Igarashi, Manurangsi,
  Schmidt-Kraepelin, and Suksompong]{Elkind2023JustifyingVoting}
Edith Elkind, Piotr Faliszewski, Ayumi Igarashi, Pasin Manurangsi, Ulrike
  Schmidt-Kraepelin, and Warut Suksompong.
\newblock {Justifying groups in multiwinner approval voting}.
\newblock \emph{Theor. Comput. Sci.}, 969:\penalty0 114039, 2023.

\bibitem[Faliszewski et~al.(2023)Faliszewski, Flis, Peters, Pierczy{\'{n}}ski,
  Skowron, Stolicki, Szufa, and Talmon]{Faliszewski2023ParticipatoryAnalysis}
Piotr Faliszewski, Jarosław Flis, Dominik Peters, Grzegorz Pierczy{\'{n}}ski,
  Piotr Skowron, Dariusz Stolicki, Stanisław Szufa, and Nimrod Talmon.
\newblock Participatory budgeting: Data, tools and analysis.
\newblock In \emph{Proceedings of the 32nd International Joint Conference on
  Artificial Intelligence, IJCAI 2023}, pages 2667--2674, 2023.

\bibitem[Fern{\'{a}}ndez et~al.(2017)Fern{\'{a}}ndez, Elkind, Lackner,
  Garc{\'{i}}a, Arias-Fisteus, Basanta-Val, and
  Skowron]{Sanchez-Fernandez2017ProportionalRepresentation}
Luis~Sánchez Fern{\'{a}}ndez, Edith Elkind, Martin Lackner,
  Norberto~Fernández Garc{\'{i}}a, Jesús Arias-Fisteus, Pablo Basanta-Val,
  and Piotr Skowron.
\newblock Proportional justified representation.
\newblock In \emph{Proceedings of the 31st AAAI Conference on Artificial
  Intelligence, AAAI 2017}, pages 670--676, 2017.

\bibitem[Godziszewski et~al.(2021)Godziszewski, Batko, Skowron, and
  Faliszewski]{Godziszewski2021AnElections}
Michał~Tomasz Godziszewski, Paweł Batko, Piotr Skowron, and Piotr
  Faliszewski.
\newblock An analysis of approval-based committee rules for {2D}-{E}uclidean
  elections.
\newblock In \emph{Proceedings of the 35th AAAI Conference on Artificial
  Intelligence, AAAI 2021}, pages 5448--5455, 2021.

\bibitem[Halpern et~al.(2023)Halpern, Kehne, Procaccia, Tucker-Foltz, and
  W{\"{u}}thrich]{Halpern2023RepresentationVotes}
Daniel Halpern, Gregory Kehne, Ariel~D Procaccia, Jamie Tucker-Foltz, and
  Manuel W{\"{u}}thrich.
\newblock Representation with incomplete votes.
\newblock In \emph{Proceedings of the 37th AAAI Conference on Artificial
  Intelligence, AAAI 2023}, pages 5657--5664, 2023.

\bibitem[Halpern et~al.(2024)Halpern, Hossain, and
  Tucker{-}Foltz]{Halpern2024ComputingVotes}
Daniel Halpern, Safwan Hossain, and Jamie Tucker{-}Foltz.
\newblock Computing voting rules with elicited incomplete votes.
\newblock In \emph{Proceedings of the 25th {ACM} Conference on Economics and
  Computation, {ACM EC} 2024}, pages 941--963, 2024.

\bibitem[Hoeffding(1963)]{Hoeffding1963ProbabilityVariables}
Wassily Hoeffding.
\newblock {Probability Inequalities for Sums of Bounded Random Variables}.
\newblock \emph{Journal of the American Statistical Association}, 58\penalty0
  (301):\penalty0 13--30, 1963.

\bibitem[Imber et~al.(2022)Imber, Israel, Brill, and
  Kimelfeld]{Imber2022Approval-BasedInformation}
Aviram Imber, Jonas Israel, Markus Brill, and Benny Kimelfeld.
\newblock Approval-based committee voting under incomplete information.
\newblock In \emph{Proceedings of the 36th AAAI Conference on Artificial
  Intelligence, AAAI 2022}, pages 5076--5083, 2022.

\bibitem[Imber et~al.(2024)Imber, Israel, Brill, Shachnai, and
  Kimelfeld]{Imber2024SpatialInformation}
Aviram Imber, Jonas Israel, Markus Brill, Hadas Shachnai, and Benny Kimelfeld.
\newblock Spatial voting with incomplete voter information.
\newblock In \emph{Proceedings of the 38th AAAI Conference on Artificial
  Intelligence, AAAI 2024}, pages 9790--9797, 2024.

\bibitem[Jha and Zick(2023)]{JhaZ23}
Tushant Jha and Yair Zick.
\newblock A learning framework for distribution-based game-theoretic solution
  concepts.
\newblock \emph{{ACM} Transactions on Economics and Computation}, 11:\penalty0
  5:1--5:23, 2023.

\bibitem[Kalayci et~al.(2025)Kalayci, Liu, and
  Kempe]{Kalayci2025FullRepresentation}
Yusuf~Hakan Kalayci, Jiasen Liu, and David Kempe.
\newblock Full proportional justified representation.
\newblock In \emph{Proceedings of the 24th International Conference on
  Autonomous Agents and Multiagent Systems, {AAMAS} 2025}, pages 1070--1078,
  2025.

\bibitem[Lackner and Skowron(2023)]{Lackner2023Multi-WinnerPreferences}
Martin Lackner and Piotr Skowron.
\newblock \emph{{Multi-Winner Voting with Approval Preferences}}.
\newblock Springer International Publishing, Cham, 2023.

\bibitem[Lewenberg et~al.(2016)Lewenberg, Bachrach, Bordeaux, and
  Kohli]{Lewenberg2016PoliticalModel}
Yoad Lewenberg, Yoram Bachrach, Lucas Bordeaux, and Pushmeet Kohli.
\newblock Political dimensionality estimation using a probabilistic graphical
  model.
\newblock In \emph{Proceedings of the 32nd Conference on Uncertainty in
  Artificial Intelligence, UAI 2016}, pages 447--456, 2016.

\bibitem[Lindeboom et~al.(2025)Lindeboom, Brehm, Grossi, and
  Murukannaiah]{Lindeboom2025AInformation}
Feline Lindeboom, Martijn Brehm, Davide Grossi, and Pradeep Murukannaiah.
\newblock A voice for minorities: diversity in approval-based committee
  elections under incomplete or inaccurate information.
\newblock Technical report, 2025.
\newblock arXiv 2506.10843.

\bibitem[Lu and Boutilier(2013)]{Lu2013Multi-WinnerPreferences}
Tyler Lu and Craig Boutilier.
\newblock Multi-winner social choice with incomplete preferences.
\newblock In \emph{Proceedings of the 23rd International Joint Conference on
  Artificial Intelligence, IJCAI 2023}, pages 263--270, 2013.

\bibitem[Mandal et~al.(2020)Mandal, Shah, and
  Woodruff]{Mandal2020OptimalVoting}
Debmalya Mandal, Nisarg Shah, and David~P Woodruff.
\newblock Optimal communication-distortion tradeoff in voting.
\newblock In \emph{Proceedings of the 21st ACM Conference on Economics and
  Computation, ACM EC 2020}, pages 795--813, 2020.

\bibitem[Meir et~al.(2014)Meir, Lev, and Rosenschein]{Meir2014AEquilibria}
Reshef Meir, Omer Lev, and Jeffrey~S Rosenschein.
\newblock {A local-dominance theory of voting equilibria}.
\newblock In \emph{Proceedings of the ACM Conference on Economics and
  Computation, ACM EC 2014}, pages 313--330. ACM, 2014.

\bibitem[Miller(2018)]{Miller2018TheVoting}
Nicholas~R. Miller.
\newblock {The spatial model of social choice and voting}.
\newblock In \emph{Handbook of Social Choice and Voting}, chapter~10, pages
  163--181. Edward Elgar Publishing Limited, 2018.

\bibitem[Peters and Lackner(2020)]{Peters2020PreferencesCircle}
Dominik Peters and Martin Lackner.
\newblock Preferences single-peaked on a circle.
\newblock \emph{Journal of Artificial Intelligence Research}, 68:\penalty0
  463--502, 2020.

\bibitem[Peters and Skowron(2020)]{Peters2020ProportionalityWelfarism}
Dominik Peters and Piotr Skowron.
\newblock Proportionality and the limits of welfarism.
\newblock In \emph{Proceedings of the 21st ACM Conference on Economics and
  Computation, ACM EC 2020}, pages 793--794, 2020.

\bibitem[Peters et~al.(2021)Peters, Pierczy{\'{n}}ski, and
  Skowron]{Peters2021ProportionalUtilities}
Dominik Peters, Grzegorz Pierczy{\'{n}}ski, and Piotr Skowron.
\newblock Proportional participatory budgeting with additive utilities.
\newblock In \emph{Proceedings of the 34th Annual Conference on Neural
  Information Processing Systems, NeurIPS 2021}, pages 12726--12737, 2021.

\bibitem[Pierczy{\'{n}}ski and
  Skowron(2022)]{Pierczynski2022Core-StableDomains}
Grzegorz Pierczy{\'{n}}ski and Piotr Skowron.
\newblock Core-stable committees under restricted domains.
\newblock In \emph{Proceedings of the 18th International Conference on Web and
  Internet Economics, WINE 2022}, pages 311--329, 2022.

\bibitem[Revel et~al.(2025)Revel, Milli, Lu, Watson{-}Daniels, and
  Nickel]{revel25}
Manon Revel, Smitha Milli, Tyler Lu, Jamelle Watson{-}Daniels, and Max Nickel.
\newblock Representative ranking for deliberation in the public sphere.
\newblock In \emph{Proceedings of the 42nd International Conference on Machine
  Learning, ICML 2025}, 2025.

\bibitem[Small et~al.(2021)Small, Bjorkegren, Erkkil{\"a}, Shaw, and
  Megill]{small2021polis}
Christopher Small, Michael Bjorkegren, Timo Erkkil{\"a}, Lynette Shaw, and
  Colin Megill.
\newblock Polis: Scaling deliberation by mapping high dimensional opinion
  spaces.
\newblock \emph{Recerca: revista de pensament i an{\`a}lisi}, 26\penalty0 (2),
  2021.

\bibitem[Szufa et~al.(2022)Szufa, Faliszewski, Janeczko, Lackner, Slinko,
  Sornat, and Talmon]{Szufa2022HowElections}
Stanisław Szufa, Piotr Faliszewski, \L{}ukasz Janeczko, Martin Lackner,
  Arkadii Slinko, Krzysztof Sornat, and Nimrod Talmon.
\newblock How to sample approval elections?
\newblock In \emph{Proceedings of the 31st International Joint Conference on
  Artificial Intelligence, IJCAI 2022}, pages 496--502, 2022.

\bibitem[Terzopoulou et~al.(2021)Terzopoulou, Karpov, and
  Obraztsova]{Terzopoulou2021RestrictedInformation}
Zoi Terzopoulou, Alexander Karpov, and Svetlana Obraztsova.
\newblock Restricted domains of dichotomous preferences with possibly
  incomplete information.
\newblock In \emph{Proceedings of the 35th AAAI Conference on Artificial
  Intelligence, AAAI 2021}, pages 5726--5733, 2021.

\bibitem[{The Computational Democracy
  Project}(2022)]{TheComputationalDemocracyProject2022TheAustria}
{The Computational Democracy Project}.
\newblock {The Klimarat in Austria}, 8 2022.
\newblock URL \url{compdemocracy.org/Case-studies/2022-Austria-Klimarat/}.

\bibitem[Zhao et~al.(2018)Zhao, Li, Wang, Kephart, Mattei, Su, and
  Xia]{Zhao2018AAggregation}
Zhibing Zhao, Haoming Li, Junming Wang, Jeffrey~O Kephart, Nicholas Mattei, Hui
  Su, and Lirong Xia.
\newblock A cost-effective framework for preference elicitation and
  aggregation.
\newblock In \emph{Proceedings of the 34th Conference on Uncertainty in
  Artificial Intelligence, UAI 2018}, pages 446--456, 2018.

\end{thebibliography}

\clearpage
\appendix
\begin{center}
    {\LARGE\bf Appendix}
\end{center}
\section{Core implies PJR+}
\label{sec:corepjr+}
\begin{proposition}
    If $W$ fails PJR+, then $W$ is not in the core.
\end{proposition}
\begin{proof}
    Suppose that committee $W$ fails PJR+. Then there exists some $\ell\in[k], G\subseteq V$ such that $|G|\geq n\ell/k$ and $A(G)\setminus W\neq \varnothing$, but $|W\cap \bigcup_{v\in G}A(v)|< \ell$. Let $c\in A(G)\setminus W$, let $T'= W\cap \bigcup_{v\in G}A(v)$ and $T=\{c\}\cup T'$. Then $|T|\leq \ell$,  but for all $v\in G$ it holds that $|W\cap A(v)|=|T'\cap A(v)|=|T\cap A(v)|-1<|T\cap A(v)|$. Hence, $W$ is not in the core.
\end{proof}
\section{Uniformisation Is Equivalent}
\label{sec:unif}
Fix a general RIV model 
${\mathcal M}=(I_t, F_t, p_t)_{t\in[\sigma]}$
and a candidate set $C\subset \bigcup_{t\in[\sigma]}I_t$.
Let ${\mathcal M}'=(I'_t, F'_t, p_t)_{t\in[\sigma]}$ be a uniform RIV model 
where for each $t\in[\sigma]$ it holds that $I_t=[t, t+1]$ and $F_t$ is the uniform distribution over $I_t$.
We define a mapping $\mu:\bigcup_{t\in[\sigma]}I_t\rightarrow[1,\sigma+1]$,  where for each $t\in[\sigma]$ and $x\in I_t$ we set
$\mu(x)=t+F_t(x)$. 

Consider a  voter $v$ with endpoints $a_v,b_v$ and $A^*(v)=[a_v,b_v]$, whose approvals are drawn from $\mathcal M$. 
Then $\mu(a_v),\mu(b_v)$ are distributed according to ${\mathcal M}'$, and, moreover, $x\in A^*(v)$ if and only if $\mu(x)\in \mu(A^*(v))$. 

This transformation enables us to state all our results for uniform RIVs. Indeed, given a general RIV $\mathcal M$, we can first map the candidate set $C$ to $\mu(C)$, and then run the algorithms on this new set of candidates assuming a uniform RIV. Whenever an algorithm needs to perform a point query
\point$(x, v)$ with  $x\in \mu(C)$ in ${\mathcal M}'$, we reverse the transformation, and perform a point query \textsc{Point}$(\mu^{-1}(x), v)$; we use a similar approach for interval queries.

Given an election $E$ where voters' preferences are drawn from $\mathcal M$, we denote by $\mu(E)$ the election where each candidate's position and all voters' endpoints are replaced with their images under $\mu$.
Then an outcome $W\subseteq \mu(C)$ of $\mu(E)$ can be converted to an outcome $\mu^{-1}(W)\subseteq C$ for $E$. We claim that $W$ satisfies PJR+/core stability for $\mu(E)$ if and only if $\mu^{-1}(W)$ satisfies PJR+/core stability for $E$.
Indeed, consider a set of voters $G\subseteq V$ with $|G|\geq n\ell/k$ with $A(G)\setminus W\neq \varnothing$. Then
$$
\mu(\mu^{-1}(W)\cap\bigcup_{v\in G}A(v))=W\cap \bigcup_{v\in G}\mu\circ A(v)
$$
and so 
\begin{align*}
|\mu^{-1}(W)\cap\bigcup_{v\in G}A(v)|&<\ell\text{ if and only if }\\ |W\cap \bigcup_{v\in G}\mu\circ A(v)|&<\ell.
\end{align*}
Therefore $W$ satisfies PJR+ in the election $\mu(E)$ if and only if $\mu^{-1}(W)$ satisfies PJR+ in $E$. 

Similarly, let $v\in V$, $T\subseteq C,W\subseteq \mu(C)$. Then for every $v\in V$ it holds that
\begin{align*}
    |A(v)\cap T|\leq& |A(v)\cap \mu^{-1}(W)|\text{ if and only if }\\ 
    |\mu(A(v)\cap T)|\leq& |\mu(A(v)\cap \mu^{-1}(W))|,
\end{align*}
and $\mu(A(v)\cap T)=\mu\circ A(v)\cap \mu(T)$, $\mu(A(v)\cap \mu^{-1}(W))=\mu\circ A(v)\cap W $, as desired.

        
\section{Proof of \texorpdfstring{Lemma \ref{lem:distlim}}{}}
\label{sec:distlimproof}
\begin{proof}
    Clearly, we have that $d^+ _{r_1}(c)+d^- _{r_2}(c)\leq 1$, and so if $r_1+r_2+1\geq \frac{k_t}{2}+1$ we get the result. So let us assume $r_1+r_2+1<\frac{k_t}{2}+1$. 
    Then it cannot be true that both $\rho(c)+r_1\geq k_t$ and $\rho(c)-r_2<1$; if that were true, then $r_1+r_2+1> k_t-\rho(c)+\rho(c)=k_t$, so $r_1+r_2\geq k_t$ which contradicts $r_1+r_2+1<\frac{k_t}{2}+1$.
    We consider these cases.
    \paragraph{Case where both $\rho(c)+r_1 < k_t$ and $\rho(c)-r_2\geq 1$;}
    \label{subsec:srivdistnice}
    Suppose that there exist at least $r_1$ marked points in $I_t$ to the right of $c$. Then let $\mu^+ _1=t+\frac{2\xi-1}{2(k_t+2)}$ be the closest marked point to the right of $c$, and $\mu^+ _h=t+\frac{2(\xi+h-1)-1}{2(k_t+2)}$ the $h$-th closest marked point to the right of $c$. Then there exists a candidate chosen for $\W_t$ that lies in $(c,\mu ^+_h +(\mu ^+_h-c)]$; otherwise $c$ would have been chosen for marked point $\mu ^+_h$. Therefore within $(c,2\mu^+ _{r_1+1}-c]$ there exist at least $r_1+1$ members of $\W_t$. 
    
    Hence 
    $\W_t[\rho(c)+r_1]\leq 2\mu^+ _{r_1+1}-c$
    and so $d^+_{r_1}(c)\leq 2(\mu^+ _{r_1+1}-c)= 2(\mu^+ _{r_1+1}-\mu^+ _{1}+\mu^+ _{1}-c)=\frac{2r_1}{k_t+2}+2(\mu^+ _{1}-c)$. By similar logic, we get that $d^- _{r_2}(c)\leq \frac{2r_2}{k_t+2}+2(c-\mu^- _{1})$ if there exist at least $r_2$ marked points to the left of $c$.

    Now suppose there there are fewer than $r_1$ marked points to the right of $c$. We have $r_1> 1$: there always exists at least one marked point to the right of $c$ in this case, since $c<t+1-\frac{3}{2(k_t+2)}$ ($c\in K_t$). Then $c\geq t+\frac{2(k_t+2-r_1)-1}{2(k_t+2)}$ and therefore $d^+ _{r_1}(c)\leq 1-\frac{2(k_t+2-r_1)-1}{2(k_t+2)}=\frac{2r_1+1}{2(k_t+2)}< \frac{2r_1}{k_t+2}+2(\mu^+ _{1}-c)$ since $r_1>1$. Again, by similar logic we get $d^- _{r_2}(c)\leq \frac{2r_2}{k_t+2}+2(c-\mu^- _{1})$ if there exist fewer than $r_2$ marked points to left of $c$.

    Hence we have
    $$
        d^+ _{r_1}(c)+d^- _{r_2}(c)\leq \frac{2r_1}{k_t+2}+2(\mu^+ _{1}-c)+\frac{2r_2}{k_t+2}+2(c-\mu^- _{1})\\=\frac{2(r_1+r_2+1)}{k_t+2} ,
    $$
    where the last equality holds because $\mu_1 ^+-\mu _1 ^-=\frac{1}{k_t+2}$.

    We now have dealt with the first case in which both $\rho(c)+r_1+1\leq k_t$ and $\rho(c)-r_2\geq 1$. Now let us look at the case when one of these conditions does not hold.
    \paragraph{Case where $\rho(c)+r_1 < k_t$ and $\rho(c)-r_2< 1$ (and vice versa)}
    Then $d_{r_2}^-(c)=c-t$ and $d_{r_1}^+(c)+d_{r_2}^-(c)=\W_t[\rho(c)+r_1] -t.$
    Note that if $c\not \in \W_t$, and there exist $r$ marked points to the left of $c$, then $\rho(c)\geq r$; otherwise there exists some marked point that has chosen a candidate to be in $\W_t$ which lies to the right of $c$, further away from the marked point than $c$, and hence $c$ should have been chosen instead, a contradiction. So then $\rho(c)\geq \floor{(k_t+2)(c-t)+\frac{1}{2}}-1$ for $c\geq t+\frac{3}{2(k+2)}$. 

    Note also then that there must exist at least $r_1$ marked points to the right of $c$.  If there are fewer than $r_1$ marked points to the right of $c$, then there are at least $k_t-(r_1-1)$ marked points to the left of $c$. Since $c\not\in \W_t$, there must exist at least $k_t-r_1+1$ members of $\W_t$ to the left of $c$, so $\rho(c)\geq k_t-r_1+1$, which contradicts $\rho(c)+r_1+1\leq k_t$.
    
    By similar argument in the previous case, we obtain $\W_t[\rho(c)+r_1]\leq 2\mu ^+ _{r_1+1}-c$. We have $\mu ^+ _{r_1+1}-c\leq (r_1+1)/(k_t+2)$ and so 
    $$
        d_{r_1}^+(c)+d_{r_2}^-(c)=\W_t[\rho(c)+r_1]-t\leq 2\mu ^+ _{r_1+1}-c-t\\=c-t+2(\mu ^+ _{r_1+1}-c)\leq c-t+2\frac{r_1+1}{k_t+2}.
    $$
    Suppose to the contrary that $c-t>\frac{2r_2}{k_t+2}$. Then
    $$
        r_2\geq \rho(c)\geq \floor{(k_t+2)(c-t)+\frac{1}{2}}-1
        \geq \floor{2r_2+\frac{1}{2}}-1=2r_2-1 , 
    $$
    which is a contradiction for $r_2>1$. So we now consider the case $\rho(c)\leq r_2\leq 1$. We have $\rho(c)\neq 0$ since $c-t\geq \frac{3}{2(k_t+2)}$ and $c\not\in \W_r$, so therefore $1=r_2=\rho(c)$ is the only case we need to be concerned with. Again, since $c\not\in \W_t$, we know that $c-t\leq \frac{5}{2(k_t+2)}$ since $\rho(c)=1$, and so $\mu_h^+=t+\frac{2h+3}{2(k_r+2)}$ and so
    $$
        2\mu^+ _{r_1+1}-c-t \leq \frac{2(r_1+1)+3}{k_t+2}-\frac{3}{2(k_t+2)}\\=\frac{2r_1+3.5}{k_t+2}< \frac{2(r_1+r_2+1)}{k_t+2}.
    $$
    Therefore $d_{r_1}^+(c)+d_{r_2}^-(c)\leq \frac{2(r_1+r_2+1)}{k_t+2}$.
    
    Using similar reasoning, we get the same result if $\rho(c)+r_1+1> k_t$ and $\rho(c)-r_2\geq 1$.
\end{proof}
\section{Proof of Lemma \texorpdfstring{\ref{lem:combinedbound}}{}}
\label{sec:combinedboundproof}
We wish to prove that for $c\in C_t\setminus \W$ it holds that 
    $\pi(\ell,c,j)\leq \ell/k-1/(4k).$ First we prove \Cref{lem:cinK} to show the bound for $c\in (C_t\setminus K_t)\setminus \W_t$, then we prove it for $c\in K_t\setminus\W_t$.
\begin{lemma}
\label{lem:cinK}
    For each $c\in (C_t\setminus K_t)\setminus \W_t$ it holds that 
    $\pi(\ell,c,j)\leq \ell/k-1/(4k).$
\end{lemma}
\begin{proof}
    We have $\Pr[c\in A(v)]\leq 2p_t\cdot\frac{3}{2(k_t+2)}$. Note that $\ell\geq 4$ implies $2p_t\cdot\frac{3}{2(k_t+2)}\leq \ell/k-1/(4k)$.
    Also, if $k_t<\ell$, then $p_t\cdot k<\ell$, and $\Pr[c\in A(v)]\leq p_t/2<\ell/(2k)\leq \ell/k-1/(4k)$, so assume now that $\ell\leq k_t$ and $\ell\leq 3$. We can assume without loss of generality that $c\leq t+\frac{3}{2(k_t+2)}$: otherwise, we can mirror the segment to get the result for $c\geq t+1-\frac{3}{2(k_t+2)}$. In this case with $\ell\leq k_t$, there exists at least $\ell$ marked points to the right of $c$. Consider the candidate $y\in \W_t$ that was selected for the marked point $p=t+\frac{2\ell+1}{2(k_t+1)}$; we have $|y-p|\leq |p-c|$ and so $|c-y|\leq|p-c|+|y-p|\leq  2(p-c)$.
    
    We also know that $|\W_t\cap(c,p+|p-c|]|\geq \ell$; each marked point in $q\in\left \{t+\frac{2\xi+1}{2(k_t+1)}:\xi\in[\ell]\right\}$ selects its candidate from $(c,q+|q-c|]$. We have 
    \begin{multline*}
        \Pr[c\in A(v)\land |A(v)\cap W_t|<\ell]\leq 2p_t(c-t)(2p-c-c)\\\leq4p_t(c-t)(p-c)
        =4p_t(c-t)\left(t+\frac{2\ell+1}{2(k_t+2)}-c\right).
    \end{multline*}
    
        We have two cases to consider.
    If $\ell=1$ or $\ell=2$, then
    $$
        4p_t(c-t)\left(t+\frac{2\ell+1}{2(k_t+2)}-c\right)\leq 4p_t\left(\frac{2\ell+1}{4(k_t+2)}\right)^2
        \\\leq \frac{(4\ell^2+4\ell+1)p_t}{4\left(k_t+2\right)^2}
        \leq \ell/k-1/4k.
   $$
    If $\ell=3$, then, recalling that $c<\frac{3}{2(k_t+2)}$, we obtain
    $$
       4p_t(c-t)\left(t+\frac{2\ell+1}{2(k_t+2)}-c\right)\leq4p_t\frac{3}{2(k_t+2)}\frac{4}{2(k_t+2)}
        \\\leq \frac{12p_t}{(k_t+2)^2}
        \leq \frac{\ell}{k}-\frac{1}{4k}, 
    $$
    since $k_t\geq \ell$ and hence we are done.
\end{proof}
\begin{proof}[Proof of \Cref{lem:combinedbound}]
    The case $c\in (C_t\setminus K_t)\setminus \W_t$ is handled by \Cref{lem:cinK}, so assume $c\in K_t\setminus \W_t$.
    If $2\ell\leq k_t+1$, we have
    $
        \pi(\ell,c,j)\leq\frac{2p_t\ell^2}{(k_t+2)^2}< \frac{2p_t\ell^2}{kp_t(2\ell+1)}
        =\frac{2\ell^2}{k(2\ell+1)}
        \leq \ell/k-1/(4k)
    $
    for $\ell\geq 1$. Now if $2\ell\geq  k_t+2$ (recalling $k_t,\ell\in\mathbb{N}$), we have $2\ell-1\geq  k_t+1> p_tk$ and so $\Pr[c\in A(v)]\leq p_t/2< l/k-1/(4k)$, and we are done.
\end{proof}
\section{Proof that \texorpdfstring{$\W$}{} provides PJR+}
\label{sec:WPJR}
    Suppose that there exists an $\ell\in[k]$ and
    a subset of voters $G\subseteq V$ with $|G|\geq n\ell/k$ such that $|\W\cap \bigcup_{v\in G} A(v)|<\ell$, but there exists a candidate $c\in A(G)\setminus \W$. Given that such a $G,c,\ell$ exists, assume that $\ell$ is the minimum such value, so $|\W\cap \bigcup_{v\in G}A(v)|=l-1$. We will show that this implies $s_{\ell,c,j}+u\geq n\ell/k$ for some $j\in \{0\}\cup[\ell-1]$, i.e., 
    \Cref{alg:PJR} does not output $\W$. Let
    {\allowdisplaybreaks
    \begin{align*}
         H=\W\cap \bigcup_{v\in G} A(v), \quad   S=\W\setminus H, \quad
         L=\{t\in H:t<c\}, \quad  j=|L|.
    \end{align*}
    }
    Since $c\in A(G)$, we know that all $A(v)$ for $v\in G$ overlap at $c$, and so therefore $\bigcup_{v\in G}A^*(v)=[\alpha,\beta]$ for some $t(c)\leq \alpha\leq\beta\leq t(c)+1 $. Therefore, $H=\W\cap [\alpha,\beta ]$. We also have $|H|=\ell-1$ since $\ell$ was minimal.

    In the iteration of the algorithm using the above $\ell,c,j$, we use $R=\W[\rho(c)-j:\rho(c)+\ell-1-j]$. We shall show that $H=R$.
    First, $L=\W[\rho(c)-j:\rho(c)]$; $L$ is the closest $j$ candidates of $H$ that are to the left of $c$, and $\W[\rho(c)-j:\rho(c)]$ is the closest $j$ candidates in $\W$ to the left of $c$. Since $H=\W\cap [\alpha,\beta ]$, these two sets are the same.
    
    We also have $H\setminus L=\W[\rho(c):\rho(c)+\ell-1-j]$. $\W[\rho(c):\rho(c)+\ell-1-j]$ is the $\ell -1-j$ candidates of $\W$ to the right of $c$ and $H\setminus L$ is the remaining $|H|-j=\ell-1-j$ candidates in $H$. Again, since $H=\W\cap [\alpha,\beta ]$, these two sets are the same, and so $H=R$, and therefore in this iteration of the algorithm, we call $\texttt{poss}(v,c,S)$ for each voter that approved of some query.
    
    Now, for each $v\in G$ that approved of some query, we have $
    \texttt{poss}(v,c,S)
    $: 
    Since $v$ approves of $c$ and none of $S$, we know that $a_v\in(\lefty{c}{S},c]$ and $b_v\in[c,\righty{c}{S})$. We also have $a_v\in(\phi_1,\phi_2]$ and $b_v\in [\phi_3,\phi_4)$, 
    where the quantities $\phi_1, \dots, \phi_4$
    are defined in \Cref{alg:possible}.
    Therefore $(\phi_1,\phi_2]\cap(\lefty{c}{S},c]\neq \varnothing$ and $[\phi_3,\phi_4)\cap [c,\righty{c}{S})\neq\varnothing$, and we have \texttt{poss}$(v,c,S)$ for all $v\in G$ who approve of a query. Any voter in $G$ that does not approve of any query gets counted in $u$ and hence $u+s_{\ell,c,j}\geq |G|\geq n\ell/k$.  Therefore, if $\W$ is output by \Cref{alg:PJR}, then it provides PJR+ for the election. 
\end{document}